\documentclass[12pt]{article}
\usepackage[margin=1.2in]{geometry}

\usepackage{amsmath,amssymb,amsfonts,graphicx,amsthm,nicefrac,mathtools,bm}
\usepackage{hyperref}
\usepackage[numbers]{natbib}
\usepackage{bbm}
\usepackage{color}
\usepackage{xspace}
\newtheorem{theorem}{Theorem}

\theoremstyle{definition}

\theoremstyle{remark}

\makeatletter
\newtheorem*{rep@theorem}{\rep@title}
\newcommand{\newreptheorem}[2]{%
\newenvironment{rep#1}[1]{%
\def\rep@title{#2 \ref{##1}}%
\begin{rep@theorem}}%
{\end{rep@theorem}}}
\makeatother
\newreptheorem{theorem}{Theorem}
\newreptheorem{lemma}{Lemma}
\newreptheorem{corollary}{Corollary}
\newreptheorem{proposition}{Proposition}

\newcommand{\iid}[0]{i.i.d.\xspace}

\newcommand{\One}[1]{{\mathbbm{1}}\left\{{#1}\right\}}

\newcommand{\PP}[1]{\mathbb{P}\left\{{#1}\right\}} 
\newcommand{\Pp}[2]{\mathbb{P}_{#1}\left\{{#2}\right\}} 
\newcommand{\EE}[1]{\mathbb{E}\left[{#1}\right]} 
\newcommand{\Ep}[2]{\mathbb{E}_{#1}\left[{#2}\right]}
\newcommand{\EEst}[2]{\mathbb{E}\left[{#1}\ \middle| \ {#2}\right]} 
\newcommand{\PPst}[2]{\mathbb{P}\left\{{#1}\ \middle| \ {#2}\right\}} 

\def\R{\mathbb{R}}
\def\Z{\mathbb{Z}}

\newcommand\independent{\protect\mathpalette{\protect\independenT}{\perp}}
\def\independenT#1#2{\mathrel{\rlap{$#1#2$}\mkern2mu{#1#2}}}

\newcommand{\eps}{\epsilon}

\newcommand{\footremember}[2]{%
    \footnote{#2}
    \newcounter{#1}
    \setcounter{#1}{\value{footnote}}%
}
\newcommand{\footrecall}[1]{%
    \footnotemark[\value{#1}]%
}

\title{The conditional permutation test for independence while controlling for confounders}
\author{\normalsize Thomas B.~Berrett\footremember{Cam}{Statistical Laboratory, University of Cambridge},
Yi Wang\footremember{Chi}{Department of Statistics, University of Chicago},
Rina Foygel Barber\footrecall{Chi} , Richard J.~Samworth\footrecall{Cam}}
\date{\normalsize \today}

\newcommand{\X}{\mathbf{X}}
\newcommand{\x}{\mathbf{x}}
\newcommand{\Y}{\mathbf{Y}}
\renewcommand{\Z}{\mathbf{Z}}

\newcommand{\Xcal}{\mathcal{X}}
\newcommand{\Ycal}{\mathcal{Y}}
\newcommand{\Zcal}{\mathcal{Z}}
\newcommand{\eqd}{\stackrel{\textnormal{d}}{=}}
\newcommand{\tv}{\textnormal{d}_{\textnormal{TV}}}
\newcommand{\kl}{\textnormal{d}_{\textnormal{KL}}}

\usepackage{algorithm}
\usepackage{algorithmic}
\usepackage{tikz,color,subcaption}
\usepackage{natbib}

\begin{document}

\maketitle

\begin{abstract}
We propose a general new method, the \emph{conditional permutation test}, for testing the conditional independence of variables $X$ and $Y$ given a potentially high-dimensional random vector $Z$ that may contain confounding factors. The proposed test permutes entries of $X$ non-uniformly, so as to respect the existing dependence between $X$ and $Z$ and thus account for the presence of these confounders. Like the conditional randomization test of \citet{candes2018panning}, our test relies on the availability of an approximation to the distribution of $X|Z$---while \citet{candes2018panning}'s test uses this estimate to draw new $X$ values, for our test we use this approximation to design an appropriate non-uniform distribution on permutations of the $X$ values already seen in the true data. We provide an efficient Markov Chain Monte Carlo sampler for the implementation of our method, and establish bounds on the Type~I error in terms of the error in the approximation of the conditional distribution of $X|Z$, finding that, for the worst case test statistic, the inflation in Type I error of the conditional permutation test is no larger than that of the conditional randomization test. We validate these theoretical results with experiments on simulated data and on the Capital Bikeshare data set.
\end{abstract}

\section{Introduction}\label{sec:intro}
Independence is a central notion in statistical model building, as well as being a foundational concept for much of statistical theory.  Originating with Francis Galton's work on correlation at the end of the 19th century~\citep{Stigler89}, many measures of dependence have been proposed, including mutual information, the Hilbert--Schmidt independence criterion, and distance covariance~\citep{Cover12,Gretton05,Szekely07}; see also~\citep{josse2013measures} for an overview.  Simultaneously, a great deal of research effort has gone into developing several different tests of independence, for example based on ranks, kernel methods, copulas, and nearest neighbours~\citep{Weihs17,Pfister18,Kojadinovic2009,Berrett17}.  Permutation tests are particularly attractive due to their simplicity and their ability to control the Type~I error (i.e., the false positive rate) without any distributional assumptions.

In practice, it is often conditional independence that is in fact of primary interest \citep{Dawid79}.  For instance, in generalized linear models for a response $Y\in\R$ regressed on a high-dimensional feature vector $(X,Z) = (X,Z^1,\dots,Z^p)\in\R^{p+1}$, the regression coefficient on feature $X$ is zero if and only if $Y$ and $X$ are conditionally independent given the remaining $p$ features, $Z = (Z^1,\dots,Z^p)$. In this paper, we will study the general problem  of testing $X\independent Y|Z$.\footnote{In the regression literature, it is more common to use the notation of regressing $Y$ on $(X^1,\dots,X^p)$, and testing whether the coefficient on feature $X^j$ is zero after controlling for the remaining features $X^{-j}=(X^1,\dots,X^{j-1},X^{j+1},\dots,X^p)$; this $X^j$ and $X^{-j}$ correspond to our $X$ and $Z$, respectively.} We are typically interested in the setting where $X$ and $Y$ are one-dimensional while $Z$ is a high-dimensional set of confounding variables that we would like to control for, but our results are not specific to this setting.

Within standard parametric regression models, conditional independence tests are well-developed; unfortunately, however, they fail to control Type~I error under model misspecification. In fact, the very recent work of \citet{shah2018hardness} has shown that, without placing some assumptions on the joint distribution of $(X,Y,Z)$, conditional testing is effectively impossible---when $(X,Y,Z)$ is continuously distributed, they prove that there is no conditional independence test that both (1) controls Type~I error over any null distribution (i.e., any distribution of $(X,Y,Z)$ with $X\independent Y|Z$), and (2) has better than random power against even one alternative hypothesis.

Our work seeks to complement this fundamental result of \citet{shah2018hardness} by demonstrating that, given some additional knowledge, namely an approximation to the conditional distribution of $X$ given $Z$, one can in fact derive conditional independence tests that are approximately valid in finite samples, and that have non-trivial power.

\subsection{Summary of contributions}
In this paper, we introduce a new method, called the conditional permutation test (CPT), which is inspired by the conditional randomization test (CRT) of \citet{candes2018panning}. The CPT modifies the standard permutation test by using available distributional information to account correctly for the confounding variables $Z$, which leads to a non-uniform distribution over the set of possible permutations $\pi$ on the $n$ observations in our data set, and restores Type~I error control.

Implementing the CPT is a challenging problem since we are sampling from a highly non-uniform distribution over the space of $n!$ permutations, but we propose a Monte Carlo sampler that yields an efficient implementation of the test. We additionally develop theoretical results examining the robustness of both the CPT and the CRT to slight errors in modeling assumptions, proving that Type~I error is only slightly inflated in both tests when our available distributional information is only approximately correct. In fact, in the worst case, Type~I error is always {\em less} inflated for the new CPT method as compared to the CRT. Our empirical results verify the greater robustness of the CPT, while maintaining comparable power in a range of scenarios.

\section{Background}\label{sec:background}
In this section, we briefly summarize several existing approaches to the problem of testing for dependence between $X$ and $Y$ in the presence of confounding variables. Before beginning, it will be helpful to define some brief notation. Throughout, we will assume that the data consists of \iid~data points $(X_i,Y_i,Z_i)\in\Xcal\times\Ycal\times\Zcal$ for $i=1,\dots,n$, and will write $\X=(X_1,\dots,X_n)$, $\Y=(Y_1,\dots,Y_n)$, and $\Z=(Z_1,\dots,Z_n)$.

\subsection{Permutation tests}
One key reason why handling conditional independence in nonparametric contexts is so challenging, is that the permutation approaches that are so effective for testing unconditional independence, $X\independent Y$, cannot be directly applied when we seek to test conditional independence, $X\independent Y|Z$. This is because it may be the case that the null hypothesis $H_0:X\independent Y|Z$ is true, but $X$ and $Y$ are highly marginally dependent due to correlation induced via each variable's dependence on $Z$. Under this null, if we sample a permutation $\pi$ of $\{1,\ldots,n\}$ uniformly at random, then the permuted data set $(X_{\pi(1)},Y_1),\dots,(X_{\pi(n)},Y_n)$ may have a very different distribution from the original data set $(X_1,Y_1),\dots,(X_n,Y_n)$, due to the confounding effect of $Z$.

In certain settings, in particular where $Z$ is categorical, there is a simple and well-known fix for this problem: we can group the observations according to their value of $Z$, and then permute within groups. For example, if $Z\in\{0,1\}$ is binary, we could draw a permutation $\pi$ that permutes the $X_i$'s within the set of indices $\{i:Z_i=0\}$, and separately permutes the $X_i$'s within the set $\{i:Z_i=1\}$. However, this strategy cannot be applied directly in the case where $Z$ is continuously distributed, or where $Z$ is discrete but with few repeated values (note that when $Z$ is high-dimensional, even if it is discrete, each observation $i$ will typically have a unique feature vector $Z_i$). In these settings, it is common to use a binning strategy, where first $Z$ is discretized to fall into finitely many bins, and then the ``permute within groups'' strategy is deployed. However, Type~I error control is no longer guaranteed, since the null hypothesis $H_0:X\independent Y|Z$ does not imply that $X\independent Y|(Z\in\text{bin $b$})$; the best we can usually hope for is that the latter statement would be approximately true under the null. Furthermore, in a high-dimensional setting, choosing these bins can itself be very challenging.

Apart from independence testing, permutation tests are also popular in other settings in which the null hypothesis is exchangeable \citep{ernst2004permutation}. Moreover, \citet{roach2018permutation} develop a theory of generalized permutation tests, primarily in the context of testing simple hypotheses, for non-exchangeable null models where the weights assigned to permutations are non-uniform. 

\subsection{The conditional randomization test}\label{sec:CRT}
The conditional randomization test (CRT), proposed by~\citet{candes2018panning}, works in a setting where no assumptions are made about the distribution of the response variable $Y$, but instead, it is assumed that the conditional distribution of $X$ given $Z$ is known. In practice, in semi-supervised learning settings where unlabeled data $(X,Z)$ are easier to obtain than labeled data $(X,Y,Z)$, it may be possible to obtain a very accurate estimate of the conditional distribution of $X|Z$, but testing for independence with $Y$ remains challenging due to limited sample size of the labeled data. \citet[Section 1.3]{candes2018panning} give examples of applications where unlabeled $(X,Z)$ data is amply available while labeled data $(X,Y,Z)$ is scarce---for example, genome-wide association studies (GWAS), where it is important to determine whether a particular genetic variant, $X$, affects a response $Y$ such as 
disease status or some other phenotype, even after controlling for the rest of the genome, encoded in $Z$. Human genome data, i.e., $(X,Z)$ data, is now plentiful,
but labeled data $(X,Y,Z)$ is expensive; if we do not know the disease status $Y$ of the individuals in previously collected samples, 
we need to obtain the $(X,Y,Z)$ samples ourselves.

Assuming then that the distribution of $X|Z$ is known (or is estimated accurately from a large sample of unlabeled data),
the CRT operates by sampling a new copy of the $X$ values in the data set. Letting $Q(\cdot|z)$ denote the distribution of $X$ given $Z=z$, conditional on $Z_1,\ldots,Z_n$, the CRT draws 
\[X^{(1)}_i \sim Q(\cdot|Z_i),\]
independently for each $i=1,\dots,n$, and independently of the observed $X_i$'s and $Y_i$'s. (In the special case where $X$ is binary, earlier work by \citet{rosenbaum1984conditional} proposed a related test, referred to as a ``conditional permutation test'' but which in fact resamples $X$ by estimating $\PPst{X=1}{Z}$ with a logistic model.)

Under the null hypothesis $H_0$ that $X\independent Y |Z$, we see that
\[\big(X|Y=y,Z=z\big) \eqd \big(X|Z=z\big) \sim Q(\cdot|z),\]
where $\eqd$ denotes equality in distribution. This means that
\[(\X^{(1)},\Y,\Z) \eqd (\X,\Y,\Z)\textnormal{ under $H_0$},\]
where $\X^{(1)}=(X^{(1)}_1,\dots,X^{(1)}_n)$. Any large differences between these two triples---for instance, if $\Y$ is highly correlated with $\X$ but not with $\X^{(1)}$---can therefore be interpreted as evidence against the null hypothesis. In order to construct a test of $H_0$, then, the CRT repeats this process $M$ times, sampling
\[\big(X^{(m)}_i|\X,\Y,\Z \big)\sim Q(\cdot|Z_i),\textnormal{ independently for $i=1,\dots,n$ and $m=1,\dots,M$}\]
to form control vectors $\X^{(1)},\dots,\X^{(M)}$. Under the null hypothesis, the triples $(\X,\Y,\Z)$, $(\X^{(1)},\Y,\Z)$, \dots, $(\X^{(M)},\Y,\Z)$ are all identically distributed; in fact, they are exchangeable. For any statistic $T=T(\X,\Y,\Z)$ that is chosen in advance (or, at least, without looking at $\X$), the random variables
\begin{equation}\label{eqn:CRT_M_copies}T(\X,\Y,\Z),T(\X^{(1)},\Y,\Z),\dots,T(\X^{(M)},\Y,\Z)\end{equation}
are therefore exchangeable as well. We can compute a p-value by ranking the value obtained from the true $\X$ vector against the values obtained from the CRT's copies:
\[p = \frac{1+\sum_{m=1}^M\One{T(\X^{(m)},\Y,\Z)\geq T(\X,\Y,\Z)}}{1+M}.\]
The exchangeability of the random variables in~\eqref{eqn:CRT_M_copies} ensures that this is a valid p-value under the null, i.e., it satisfies $\PP{p\leq \alpha}\leq \alpha$ for all $\alpha\in[0,1]$ if the null hypothesis $H_0$ is true. 

The ``model-X knockoffs'' framework of \citet{candes2018panning} also extends the CRT technique to the high-dimensional variable selection setting, where each of $p$ features is tested in turn for conditional independence with the response $Y$, with the goal of false discovery rate control. In this framework, only a single copy of each feature is created. The robustness of the model-X knockoffs method, with respect to errors in the conditional distributions used to construct the knockoff copies of each feature (analogous to the $\X^{(m)}$'s above), was studied by \citet{barber2018robust}.

\subsection{Other tests of conditional independence}\label{sec:lit_review}
Before introducing our new work, we give a brief overview of some additional conditional independence testing methods proposed in the literature. Many methods assume some parametric model for the response $Y$, such as a linear model, $Y = \alpha X+\beta^\top Z + \textnormal{(noise)}$, in which case the problem reduces to testing whether $\alpha=0$. This can be tested by, for instance, computing an estimate $\widehat{\beta}$ and testing whether the residual $Y-\widehat{\beta}^\top Z$ is correlated with $X$. \citet{belloni2014inference} propose a variant on this approach, which assumes approximate linear models for both $Y$ and $X$. Their method regresses both $X$ and $Y$ on $Z$, then tests for correlation between the two resulting residual vectors; this ``double regression'' offers superior performance by removing much of the bias coming from errors in estimating the effect of $Z$. \citet{shah2018hardness} consider a more general double regression framework, assuming that the conditional means $\EEst{X}{Z=z}$ and $\EEst{Y}{Z=z}$ can be estimated at a sufficiently fast rate.

Away from the regression setting, many proposed methods are based on using kernel representations or low-dimensional projections of the data. Tests based on embedding the data into reproducing kernel Hilbert spaces are studied in, for example, \citet{Fukumizu2008,Zhang2011} and \citet{strobl2017approximate}. Other works use permutations of the data, including \citet{doran2014permutation} and \citet{sen2017model}, where the methods have the flavor of binning $Z$ and then permuting within groups. \citet{bergsma2004testing}, \citet{song2009testing} and \citet{veraverbeke2011estimation} study copula methods
for testing conditional independence. There is also a large literature on extending measures of marginal independence to the conditional setting, including partial distance covariance \citep{szekely2014partial}; conditional mutual information \citep{Runge2018}; characteristic functions \citep{su2007consistent}; Hellinger distances \citep{su2008nonparametric}; and smoothed empirical likelihoods \citep{su2014testing}.

A related problem is that of testing the null hypothesis that a certain treatment has no effect in a randomized experiment. In the treatment effects literature it is common to calculate p-values by comparing a test statistic to null statistics based on randomly reassigning treatments in the data. However, in some situations, uniformly random reassignment is inappropriate, and does not result in valid p-values, due to the presence of some underlying structure; see \citet{athey2018exact} for network dependence and \citet{hennessy2016conditional} for covariate imbalance. In such cases it is sometimes possible to develop non-uniform randomization schemes that result in valid p-values, as with the CPT and the CRT.

\section{The conditional permutation test (CPT)}\label{sec:CPT}

Recall that the conditional randomization test (CRT)~\cite{candes2018panning} creates copies $\X^{(m)}$ of the vector $\X$ sampled under the null hypothesis that $X\independent Y|Z$, by drawing
\begin{equation}\label{eqn:CRT_sample}\X^{(m)}|\X,\Y,\Z\sim Q^n(\cdot| \Z),\textnormal{ independently for $m=1,\dots,M$,}\end{equation}
where we define $Q^n(\cdot| \Z) := Q(\cdot| Z_1)\times \dots \times Q(\cdot| Z_n)$. This mechanism creates copies $\X^{(1)},\dots,\X^{(M)}$ that are exchangeable with the original vector $\X$ under the null hypothesis that $X\independent Y |Z$.

Our proposed method, the conditional permutation test (CPT), is a variant on the CRT, with $\X^{(1)},\ldots,\X^{(M)}$ drawn as in~\eqref{eqn:CRT_sample} but under the constraint that each $\X^{(m)}$ must be a permutation of the original vector $\X$.  Once we have drawn $\X^{(1)},\dots,\X^{(M)}$, they will then be used exactly as for the CRT---given some predefined statistic $T = T(\X,\Y,\Z)$, our p-value is given by
\begin{equation}\label{eqn:pval_CPT}
p =  \frac{1 + \sum_{m=1}^M \One{T(\X^{(m)},\Y,\Z)\geq T(\X,\Y,\Z)}}{1+M}.
 \end{equation}
All that remains, then, is to specify how these permuted copies $\X^{(m)}$ will be drawn.

In order to draw the $\X^{(m)}$'s, we first need to define some notation. Let $\mathcal{S}_n$ denote the set of permutations on the indices $\{1,\dots,n\}$.
Given any vector $\x=(x_1,\dots,x_n)$ and any permutation $\pi\in\mathcal{S}_n$, define
$\x_\pi = (x_{\pi(1)},\dots,x_{\pi(n)})$,
i.e., the vector $\x$ with its entries reordered according to the permutation $\pi$. 

The CPT copies $\X^{(1)},\dots,\X^{(M)}$ are then drawn as follows:
after observing $\X,\Y,\Z$, we draw $M$ permutations $\pi^{(1)},\dots,\pi^{(M)}$ according to the conditional distribution of $X|Z$, and then apply these permutations to $\X$.
Specifically, let
\begin{equation}\label{eqn:CPT_sample}\X^{(m)} = \X_{\pi^{(m)}}\text{ \ where \ }\PPst{\pi^{(m)} = \pi}{\X,\Y,\Z} = \frac{q^n(\X_\pi|\Z)}{\sum_{\pi'\in\mathcal{S}_n} q^n(\X_{\pi'}|\Z)}.\end{equation}
Here we let $q(\cdot|z)$ be the density of the distribution $Q(\cdot|z)$ (i.e., $q(\cdot|z)$ is the conditional density of $X$ given $Z=z$), with respect to some base measure $\nu$ on $\Xcal$ that does not depend on $z$. We write $q^n(\cdot | \Z):=q(\cdot|Z_1)\cdot\dots\cdot q(\cdot|Z_n)$ to denote the product density. Note that we are not assuming a continuous distribution necessarily; the base measure may be discrete, allowing $X$ to be discrete as well.

Why is this the right distribution for drawing the permuted copies $\X^{(1)},\dots,\X^{(M)}$? To understand this,
it is helpful to consider a different formulation of the permutation scheme.
Let $\X_{()} = (X_{(1)},\dots,X_{(n)})$ be the order statistics of the list of values $\X=(X_1,\dots,X_n)$.\footnote{In
the setting where $\Xcal=\R$, we can of course use the usual ordering on $\R$. 
In the general case we can simply take an arbitrary total ordering on $\Xcal$; 
 the choice of ordering is irrelevant as its only role is to allow us to observe the set of values of $\X$ without knowing which one corresponds
to which data point.} Define also $\X_{(\pi)} = (X_{(\pi(1))},\dots,X_{(\pi(n))})$
for each $\pi\in\mathcal{S}_n$, and
let $\Pi\in\mathcal{S}_n$ be the permutation
given by the ranks of the true observed vector $\X$, so that $\X = \X_{(\Pi)}$. In other words,
$\X_{()}$ gives the order statistics of $\X$, and $\Pi$ reveals the ranks; together these two pieces of information are sufficient 
to reconstruct $\X$.\footnote{If the unlabeled values $X_{(i)}$ are not unique, 
then formally, we define $\Pi$ by choosing it uniformly at random from the set of all permutations that satisfy this condition.}  

Under the null hypothesis that $X\independent Y | Z$, we can verify that the distribution of the true ranks $\Pi$, conditional on $\Y,\Z$ as well as on the order statistics $\X_{()}$, is given
by
\begin{equation}\label{eqn:distrib_Pi}\PPst{\Pi = \pi}{\X_{()},\Y,\Z} =\frac{q^n(\X_{(\pi)}|\Z)}{\sum_{\pi'\in\mathcal{S}_n} q^n(\X_{(\pi')}|\Z)}.\end{equation}
Furthermore, examining the definition~\eqref{eqn:CPT_sample} of the CPT copies $\X^{(1)},\dots,\X^{(M)}$, we
can see that the CPT can equivalently be defined by
\begin{equation}\label{eqn:CPT_sample_alt}\X^{(m)} = \X_{(\Pi^{(m)})}\text{ \ where \ $\Pi^{(m)}|\X_{()},\Y,\Z$ is drawn from~\eqref{eqn:distrib_Pi}}.\end{equation}
In fact, comparing with~\eqref{eqn:CPT_sample}, we see that $\Pi^{(m)} = \Pi \circ \pi^{(m)}$.

The following theorem formalizes the above intuition, and verifies that this procedure yields a valid test of $H_0$.
\begin{theorem}\label{thm:CPT_exact}
Assume that $H_0:X\independent Y|Z$ is true, and that the conditional distribution of $X|Z$ is given by $Q(\cdot|Z)$.
Suppose that $\X^{(1)},\dots,\X^{(M)}$ are drawn i.i.d.~from the CPT sampling scheme given in~\eqref{eqn:CPT_sample}.
Then the $M+1$ triples
\[(\X,\Y,\Z), \ (\X^{(1)},\Y,\Z), \ \dots, \ (\X^{(M)},\Y,\Z)\]
are exchangeable. In particular, this implies that for any statistic $T:\Xcal^n\times\Ycal^n\times\Zcal^n\rightarrow\R$, the p-value defined in~\eqref{eqn:pval_CPT} is valid, satisfying $\PP{p\leq \alpha}\leq \alpha$ for any desired Type~I error rate $\alpha\in[0,1]$ when $H_0$ is true.
\end{theorem}
\begin{proof}[Proof of Theorem~\ref{thm:CPT_exact}]
Our work above verified that, under $H_0$, the true data vector $\X$ and the CPT copies $\X^{(1)},\dots,\X^{(M)}$
are permutations of $\X_{()}$ obtained via i.i.d.~draws from~\eqref{eqn:distrib_Pi},
conditional on $\X_{()},\Y,\Z$.
Therefore, after marginalizing over $\X_{()},\Y,\Z$, the $M+1$ triples $(\X,\Y,\Z)$, $(\X^{(1)},\Y,\Z)$, \dots, $(\X^{(M)},\Y,\Z)$ are exchangeable.
\end{proof}

\subsection{Comparing the CPT and CRT}\label{sec:compare_CPT_CRT}
To compare the construction of the copies $\X^{(1)},\dots,\X^{(M)}$
in each of the two methods, for the CPT, the copies $\X^{(1)},\dots,\X^{(M)}$ are 
i.i.d.~draws from the null distribution of $\X$, conditional on $\X_{()},\Y,\Z$.
In comparison, the CRT copies defined in~\eqref{eqn:CRT_sample}  are i.i.d.~draws from the null distribution of $\X$ conditioned
on $\Y,\Z$---but without conditioning on $\X_{()}$.

Each of these two constructions
yields a valid test if the distribution $Q(\cdot|z)$, used to draw the (resampled or permuted) copies $\X^{(m)}$, is correct---that is,
if we know the true conditional distribution of $X|Z$. This result is proved in
Theorem~\ref{thm:CPT_exact} above for the CPT, while the analogous result for the CRT is proved in \citet[Lemma 4.1]{candes2018panning}. 
However, if the null hypothesis is {\em not} true, which method might be more sensitive and better able to detect
a non-null? Furthermore, what might occur for these two methods
if $Q(\cdot|z)$ is not exactly correct?
We next explore the difference between the two methods in greater depth to begin to address these questions.

\paragraph{Use of marginal distribution of $\X$} In terms of how the tests are run, the difference between the CPT and CRT can be described as follows: while both tests use the (true or estimated) conditional distribution $Q(\cdot|Z)$, the CPT additionally uses the marginal distribution of the observed data vector $\X$, by observing its unlabeled values $\X_{()}$. Intuitively, using this additional information can in some cases make the copies $\X^{(m)}$ more similar to the original $\X$, than for the CRT. Therefore, the CPT may be somewhat less likely to reject $H_0$, which could lead to lower Type~I error if $H_0$ is true, or reduced power to detect when $H_0$ is false. In Section~\ref{sec:robust}, we will develop theory to examine the two tests' robustness to errors in estimating the conditional distribution $Q(\cdot|Z)$, and we will compare the tests in terms of both Type~I error and power in experiments in Section~\ref{sec:empirical}. 

\paragraph{Invariance to base measure} Since the CPT works only over permutations of the same set of $X$ values, it follows that it is invariant to changes in the base measure on $\Xcal$. To make this concrete, suppose that $q_1(\cdot|z)$ is another conditional density, with the property that there exist functions $h(\cdot),c(\cdot)$ such that $q_1(x|z) = q(x|z)h(x)c(z)$ for all $x\in\Xcal$ and all $z\in\Zcal$. (Here we can think of $h(x)$ as changing the base measure on $\Xcal$, while $c(z)$ adjusts the normalizing constants as needed.)

If this is the case, then running the CPT with $q_1$ in place of $q$ will have no effect on the outcome---this is because we can calculate
\[q_1^n(\X_{\pi}|\Z) = \prod_{i=1}^n q(X_{\pi(i)}|Z_i)h(X_{\pi(i)})c(Z_i) =  q^n(\X_{\pi}|\Z) \cdot \prod_{i=1}^n h(X_i)c(Z_i).\]
The first term, $q^n(\X_{\pi}|\Z)$, is the same as for the CPT run with conditional density $q$, while the second term, $\prod_{i=1}^n h(X_i)c(Z_i)$, does not depend on the permutation $\pi$ and therefore does not affect the resulting distribution of the sampled permutations. In other words, the CPT sampling distribution given in~\eqref{eqn:CPT_sample}
is unchanged if we replace $q$ with $q_1$.

This means that the CPT is a valid test, i.e., the result of Theorem~\ref{thm:CPT_exact} holds, even if the conditional density $q(\cdot|z)$ is correct only up to a change in base measure---that is, Theorem~\ref{thm:CPT_exact} holds whenever the conditional distribution $Q(\cdot|Z)$ has a density of the form $q(x|z)h(x)c(z)$, for some functions $h(\cdot),c(\cdot)$. Indeed, in some settings, it may be substantially simpler to estimate the conditional density only up to base measure---for instance, we can consider a semiparametric model with a conditional density of the form $\exp\{x\cdot z^\top \theta - f(x) - g(z)\}$, in which case the CPT would only need to estimate the parametric component $\theta$. In contrast, running the CRT requires being able to sample from the conditional distribution $Q(\cdot|Z)$, so we would need to approximate the full conditional density.
\section{Sampling algorithms for the CPT}\label{sec:MH}
In order to run the CPT, we need to be able to sample permutations $\Pi^{(1)},\dots,\Pi^{(M)}$
 from the distribution given in~\eqref{eqn:CPT_sample}. We now turn to the problem of generating such samples efficiently. 

One simple approach would be to run a Metropolis--Hastings algorithm with a proposal distribution that, from a current state $\pi$, draws its proposed permutation $\pi'$ uniformly at random. For even a moderate $n$, however, the acceptance odds ratio
\begin{equation}\label{eqn:MH_unif} \frac{q^n(\X_{\pi'}|\Z)}{q^n(\X_\pi|\Z)} = \frac{\prod_{i=1}^n q(X_{\pi'(i)}|Z_i)}{\prod_{i=1}^n q(X_{\pi(i)}|Z_i)}\end{equation}
 will be extremely low for nearly all permutations $\pi'$ (unless, of course, the dependence of $X$ on $Z$ is very weak). In other words, a uniformly drawn permutation $\pi'$ is not likely to lead to a plausible vector of $X$ values, leading to slow mixing times.

As a second attempt, we can consider a different proposal distribution: from the current state $\pi$, we propose the permutation $\pi' = \pi\circ  \sigma_{ij}$, where $\sigma_{ij}$ is the permutation that swaps indices $i$ and $j$, which are drawn at random. The acceptance odds ratio~\eqref{eqn:MH_unif} now simplifies to
\begin{equation}\label{eqn:MH_pair} \frac{q(X_{\pi(j)}|Z_i)\cdot q(X_{\pi(i)}|Z_j)}{q(X_{\pi(i)}|Z_i)\cdot q(X_{\pi(j)}|Z_j)}.\end{equation}
The probability of accepting a swap will now be reasonably high; however, each step can only alter two of the $n$ indices, again leading to slow mixing times.

\subsection{A parallelized pairwise sampler}

To address these issues, we propose a parallelized version of this pairwise algorithm. At each step, we first draw $\lfloor n/2\rfloor$ disjoint pairs of indices from $\{1,\ldots,n\}$. Next, independently and in parallel for each pair, we decide whether or not to swap this pair $(i,j)$, according to the odds ratio~\eqref{eqn:MH_pair}. This sampler is defined formally in Algorithm~\ref{alg:MC_parallel}. For ease of our theoretical analysis, we will work with the order statistics $\X_{()}$, rather than the original ordered vector $\X$, 
in our sampler; this difference is only in the notation, i.e., the algorithm can equivalently be implemented with $\X$ in place of $\X_{()}$.
\begin{algorithm}[t]
\caption{Parallelized pairwise sampler for the CPT}
\label{alg:MC_parallel}
\begin{algorithmic}
\STATE \textbf{Input:} Initial permutation $\Pi^{[0]}$, integer $S\geq 1$.\smallskip
\FOR{$s=1,2,\dots,S$}
\STATE Sample uniformly without replacement from $\{1,\ldots,n\}$ to obtain disjoint pairs 
\[(i_{s,1},j_{s,1}),\dots,(i_{s,\lfloor n/2\rfloor},j_{s,\lfloor n/2\rfloor}).\]\vspace{-.2in}
\STATE Draw independent Bernoulli variables $B_{s,1},\dots,B_{s,\lfloor n/2\rfloor}$ with odds ratios
\begin{equation}\label{eqn:MC_oddsratio}
\frac{\PP{B_{s,k}=1} }{\PP{B_{s,k}=0}}= \frac{q(X_{(\Pi^{[s-1]}(j_{s,k}))}|Z_{i_{s,k}})\cdot q(X_{(\Pi^{[s-1]}(i_{s,k}))}|Z_{j_{s,k}})}{q(X_{(\Pi^{[s-1]}(i_{s,k}))}|Z_{i_{s,k}})\cdot q(X_{(\Pi^{[s-1]}(j_{s,k}))}|Z_{j_{s,k}})}.\end{equation}
\STATE Define $\Pi^{[s]}$ by swapping $\Pi^{[s-1]}(i_{s,k})$ and $\Pi^{[s-1]}(j_{s,k})$ for each $k$ with $B_{s,k}=1$.
\ENDFOR
\end{algorithmic}
\end{algorithm}

The next theorem verifies that the resulting Markov chain yields the desired stationary distribution.
 (The proof of this theorem, and all remaining proofs, are given in Appendix~\ref{app:proofs}.)
\begin{theorem}\label{thm:MC} For every initial permutation $\Pi^{[0]}$, 
the distribution~\eqref{eqn:distrib_Pi}  of the 
permutation $\Pi$ conditional on $\X_{()},\Y,\Z$
is a stationary distribution of the Markov chain defined in Algorithm~\ref{alg:MC_parallel}. If additionally $q(x|z)>0$ for all $x\in\Xcal$ and all $z\in\Zcal$,
then it is the unique stationary distribution.
\end{theorem}
This result justifies the thought that, if Algorithm~\ref{alg:MC_parallel} is run for a sufficient number of steps $S$, then the resulting 
copy $\X_{(\Pi^{[S]})}$ acts as an appropriate control for $\X$ in testing conditional independence.
In fact, though, we can make a much stronger statement---since
the original permutation $\Pi$ also follows the 
distribution~\eqref{eqn:distrib_Pi} conditional on $\X_{()},\Y,\Z$ under the null, this means that
by initializing Algorithm~\ref{alg:MC_parallel} at $\Pi^{[0]}=\Pi$ (that is, at the original data vector $\X$),
we are initializing with a draw from the stationary distribution. Therefore $\X^{[S]} = \X_{(\Pi^{[S]})} $ is a draw from the target
distribution at any $S$, and is a valid control for $\X$ even if the number of steps $S$ is small.
Of course, if $S$ is too small, then the control copy will be too similar to the original data vector $\X$, and our power
to reject the null will be low; we explore this empirically in Section~\ref{sec:empirical}, and will see that the sampler
mixes well at even a moderate $S$ (e.g., in our experiments, we used $S=50$).

In practice, we want to draw $M$ copies, $\X^{(m)}$ for $m=1,\dots,M$, and we need to ensure that the original data $\X$ and each of the $M$ permutations $\X^{(m)}$ are all exchangeable with each other. If we sample the permuted vectors $\X^{(1)},\dots,\X^{(M)}$ sequentially, by running Algorithm~\ref{alg:MC_parallel} for $S\cdot M$ steps and extracting one copy $\X^{(m)}$ after each round of $S$ steps, then we would not achieve exchangeability, since there would be some correlation between adjacent copies in this sequence. (Of course, in practice, if the number of steps $S$ is chosen to be large, then the violation of exchangeability would be very mild.)

Instead, we can construct an exchangeable sampling mechanism with the following algorithm:
\begin{algorithm}[H]
\caption{Exchangeable sampler for multiple draws from the CPT}
\label{alg:MC_exchangeable}
\begin{algorithmic}
\STATE \textbf{Input:} Initial permutation $\Pi_{\mathrm{init}}$ and integer $S\geq 1$.\smallskip
\STATE Define $\Pi_{\sharp}$ by running Algorithm~\ref{alg:MC_parallel} initialized at $\Pi^{[0]} = \Pi_{\mathrm{init}}$ for $S$ steps.
\FOR{$m=1,\dots,M$ (independently for each $m$)}
\STATE Define $\Pi^{(m)}$ by running Algorithm~\ref{alg:MC_parallel} initialized at $\Pi^{[0]} = \Pi_{\sharp}$ for $S$ steps.
\ENDFOR
\end{algorithmic}
\end{algorithm}

Algorithm~\ref{alg:MC_exchangeable} provides an exchangeable sampling mechanism, since the permutation $\Pi_{\sharp}$ is at the ``center'', lying $S$ steps away from each of the permutations $\Pi,\Pi^{(1)},\dots,\Pi^{(M)}$. The following result verifies exchangeability:
\begin{theorem}\label{thm:MC_exchangeable} Let $\X_{()}$ and $\Pi$ be the order statistics and ranks of $\X$, as defined previously, so that $\X = \X_{(\Pi)}$. Let $\Pi^{(1)},\dots,\Pi^{(M)}$ be the output of Algorithm~\ref{alg:MC_exchangeable}, when initialized at $\Pi_{\mathrm{init}}=\Pi$, 
and let $\X^{(m)} = \X_{(\Pi^{(m)})}$ for each $m=1,\dots,M$. Assume that the null hypothesis that $X\independent Y|Z$ holds, and the conditional distribution of $X|Z$ is given by $Q(\cdot|Z)$, so that the distribution of $\Pi$ conditional on $\X_{()},\Y,\Z$ is given by~\eqref{eqn:distrib_Pi}.  Then the triples $(\X,\Y,\Z),(\X^{(1)},\Y,\Z),\dots,(\X^{(M)},\Y,\Z)$ are exchangeable.
\end{theorem}
This result ensures that the results of Theorem~\ref{thm:CPT_exact} hold when the permuted vectors $\X^{(1)},\dots,\X^{(M)}$ are obtained via the exchangeable sampler.

\section{Robustness of the CPT and CRT}\label{sec:robust}
We next consider whether the CPT and CRT, based on resampling $X$ from a known or estimated conditional distribution given $Z$, are robust to slight errors in this distribution. Suppose that the conditional distribution $Q(\cdot|Z)$ that we use for sampling when running the CPT or CRT is only an approximation to the true conditional, denoted by $Q_\star(\cdot| Z)$. In this section we provide bounds on the excess Type~I error of the CPT and CRT as a function of the difference between the true conditional $Q_\star$ and its approximation $Q$. Throughout, we will assume that the statistic $T:\Xcal^n\times\Ycal^n\times\Zcal^n\rightarrow\R$ used in the test, as well as the approximation $Q$ to the conditional distribution, are chosen independently of $\X,\Y$. For instance, in many applications, we may have access to unlabeled data, i.e., draws of $(X,Z)$ without $Y$, which we can use to construct an estimate $Q$.

Our first result demonstrates that, conditional on $\Y,\Z$, the excess Type~I error of both the CPT and the CRT is bounded by the total variation distance between $Q_\star$ and $Q$. (For any two distributions $Q_1,Q_2$ defined on the same probability space, the total variation distance is defined as $\tv(Q_1,Q_2) = \sup_A |Q_1(A) - Q_2(A)|$, where the supremum is taken over all measurable sets.)
\begin{theorem}\label{thm:robust}
Assume that $H_0:X\independent Y|Z$ is true, and that the conditional distribution of $X|Z$ is given by $Q_\star(\cdot|Z)$. For a fixed integer $M\geq 1$, let $\X^{(1)},\dots,\X^{(M)}$ be copies of $\X$ generated either from the CRT~\eqref{eqn:CRT_sample}, from the CPT~\eqref{eqn:CPT_sample},
or from the exchangeable sampler  for the CPT (Algorithm~\ref{alg:MC_exchangeable}) with any fixed parameter $S\geq1$, 
using an estimate $Q$ of the true conditional distribution $Q_\star$.

Then, for any desired Type~I error rate $\alpha\in[0,1]$, 
\[\PPst{p\leq \alpha}{\Y,\Z} \leq \alpha + \tv\big(Q_\star^n(\cdot|\Z),Q^n(\cdot|\Z)\big),\]
where $p$ is the p-value computed in~\eqref{eqn:pval_CPT}, and the probability is taken with respect to the distribution of $\X,\X^{(1)},\dots,\X^{(M)}$ conditional on $\Y,\Z$.
\end{theorem}
Of course, we can also bound the Type I error rate unconditionally, with
\[\PP{p\leq \alpha} \leq \alpha + \EE{\tv\big(Q_\star^n(\cdot|\Z),Q^n(\cdot|\Z)\big)},\]
which we obtain from the result above by marginalizing over $\Y,\Z$. 

This result ensures that, if $Q$ is a good approximation to $Q_\star$, then both the CPT and CRT will have at most a mild increase in their Type~I error. Of course, Theorem~\ref{thm:robust} is a worst-case result, proved with respect to an arbitrary statistic $T$ which may be chosen adversarially so as to be maximally sensitive to errors in estimating the true conditional distribution $Q_\star$. In practice, we might expect that the simple statistics $T$ that we would most often use, such as correlation between $\X$ and $\Y$, could be more robust to errors than the theorem suggests.

While Theorem~\ref{thm:robust} provides an upper bound on the Type~I error for both the CPT and the CRT, we do not yet have a comparison between the two. The following theorem proves that, for the case of the CRT, the upper bound is in fact tight when the number of copies $\X^{(1)},\dots,\X^{(M)}$ is large:
\begin{theorem}\label{thm:robust_lowerbd}
Under the setting and assumptions of Theorem~\ref{thm:robust}, there exists a statistic $T:\Xcal^n\times\Ycal^n\times\Zcal^n\rightarrow\R$ such that, for the CRT,\footnote{To be more precise with the constant, we can replace $0.5(1+o(1))$ with $2.5$ for any $M\geq 2$.}
\[\sup_{\alpha\in[0,1]} \bigg(\PPst{p\leq \alpha}{\Y,\Z} - \alpha \bigg)\geq \tv\big(Q_\star^n(\cdot|\Z),Q^n(\cdot|\Z)\big) - 0.5(1+o(1))\sqrt{\frac{\log(M)}{M}}\]
as $M \rightarrow \infty$.
\end{theorem}
In other words, if we use the statistic $T$ that is best able to detect errors in our conditional distribution, and choose $\alpha$ adversarially, then the excess Type~I error of the CRT is exactly equal to $ \tv\big(Q_\star^n(\cdot|\Z),Q^n(\cdot|\Z)\big)$ (up to a vanishing factor), and therefore is at least as high as that of the CPT under {\em any} statistic. 

Unlike for the CRT, we have found that there is no simple characterization of the worst-case scenario for the CPT.
In particular, for some specially constructed distributions on $(X,Y,Z)$, we can show that the CPT achieves the same lower bound as given in Theorem~\ref{thm:robust_lowerbd}
for the CRT (again, under a worst-case choice of the statistic $T$), but for other joint distributions on $(X,Y,Z)$ we can verify that the CPT cannot
achieve this error rate. In particular, since the CPT is invariant to the base measure (as discussed in Section~\ref{sec:compare_CPT_CRT}),
if $Q(\cdot|z)$ is correct up to the base measure, then the excess Type I error of CRT may be as large as $\tv\big(Q_\star^n(\cdot|\Z),Q^n(\cdot|\Z)\big)$
while the CPT is guaranteed to control Type I error at level $\alpha$.

It is important to note that the lower bound for the CRT in Theorem~\ref{thm:robust_lowerbd} applies only to a specific worst-case statistic
$T$, and does not guarantee that the excess error of the CRT will bound that of the CPT when both tests use some other statistic $T$.
However, in Section~\ref{sec:empirical} we will see that empirically, the CPT often yields a far lower Type~I error than the CRT in simulations.
Thus, we interpret Theorem~\ref{thm:robust_lowerbd} as giving us a partial theoretical understanding of this phenomenon,
since it only addresses the worst-case statistic.

\subsection{When is the total variation distance small?}
In order for Theorem~\ref{thm:robust} to have practical implications, we need to verify that there are settings where,
although the true distribution $Q_\star$ of $X|Z$ is unknown, it can be estimated to high accuracy, with $\tv\big(Q_\star^n(\cdot|\Z),Q^n(\cdot|\Z)\big) = o_p(1)$
(so that excess Type I error is guaranteed to be small). As discussed in Section~\ref{sec:CRT},
in many applications we may have a large unlabeled data set, say $(X^{\textnormal{unlab}}_i,Z^{\textnormal{unlab}}_i),i=1,\dots,N$,
with which we can compute an estimate $Q$ of $Q_\star$. 
(In fact, as discussed by \citet{barber2019construction} in the setting of model-X knockoffs, the unlabeled data set does not need
to have the same distribution over $(X,Z)$ as the labeled data, as long as the conditional distribution of $X|Z$ is the same.)

In this section, we briefly sketch two settings
where, given a large unlabeled sample size $N$, our estimate $Q$ is likely to satisfy $\tv\big(Q_\star^n(\cdot|\Z),Q^n(\cdot|\Z)\big) = o_p(1)$.
Our results here are stated informally, with no technical details, since we aim only to give intuition
for the settings where Theorem~\ref{thm:robust} is useful.

\paragraph{Parametric setting}
We will use Pinsker's inequality relating total variation distance to the Kullback--Leibler divergence, namely,
\[\tv^2\big(Q_\star^n(\cdot|\Z),Q^n(\cdot|\Z)\big)  \leq \frac{1}{2}\kl\big(Q_\star^n(\cdot|\Z),Q^n(\cdot|\Z)\big)  = \frac{1}{2}\sum_{i=1}^n \kl\Big(Q_\star(\cdot|Z_i),Q(\cdot|Z_i)\Big).\]
It is therefore sufficient to show that $\sum_{i=1}^n \kl\Big(Q_\star(\cdot|Z_i),Q(\cdot|Z_i)\Big)= o_p(1)$.

In fact, if the true conditional distribution $Q_\star(\cdot|z)$ belongs to a parametric family, then this will typically hold
whenever the unlabeled sample size satisfies $N \gg n\cdot k$, where $k$ is the number of parameters defining the models in the family. 
Specifically, we can think of a setting where  $Q_\star(\cdot|z)$ has density $f_{\theta_\star}(\cdot|z)$, where $\theta_\star\in\R^k$ is the unknown parameter vector while 
the family of densities $f_\theta(\cdot|z)$ is known. For example, suppose that $\Zcal = \R^{k-1}$, and  the conditional distribution of $X|Z$ is given by
\[X|Z=z \sim \mathcal{N}(z^\top\beta_\star, \sigma_\star^2).\]
Then the unknown parameters are $\theta_\star = (\beta_\star,\sigma_\star^2)$ and standard least squares theory allows us to produce independent (maximum likelihood) estimates $\widehat\beta,\widehat\sigma^2$ satisfying
\[
	\widehat\beta \sim N_{k-1} \bigl( \beta_\star, \sigma_\star^2(\Z_{\textnormal{unlab}}^\top \Z_{\textnormal{unlab}})^{-1} \bigr), \quad \widehat\sigma^2 \sim \frac{\sigma_\star^2}{N} \chi_{N-k+1}^2, 
\]
where $\Z_{\textnormal{unlab}}$ is the $N \times (k-1)$ matrix with $i$th row $Z_i^{\textnormal{unlab}}$.  Thus, for any $z \in \mathcal{Z}$,
\begin{align*}
	\kl\Big(Q_\star(\cdot|z),Q(\cdot|z)\Big)  &= \kl\Big(\mathcal{N}(z^\top\beta_\star,\sigma_\star^2),\mathcal{N}(z^\top\widehat\beta,\widehat\sigma^2)\Big) \\
	& = \log \frac{\widehat\sigma}{\sigma_\star} + \frac{\sigma_\star^2 }{2\widehat\sigma^2 } - \frac{1}{2} + \frac{(z^\top\widehat\beta-z^\top\beta_\star)^2}{2\widehat\sigma^2} 
                                                    = O_p \biggl( \frac{1 + \|z\|^2}{N} \biggr)
\end{align*}
under mild conditions on the distribution of $Z$. Putting everything together, if $Z$ has a finite second moment we then have
\[\tv\big(Q_\star^n(\cdot|\Z),Q^n(\cdot|\Z)\big)  = O_p \biggl( \sqrt{n \cdot \frac{k}{N}} \biggr),\]
which is vanishing as long as the unlabeled sample size satisfies $N \gg n\cdot k$.

\paragraph{Nonparametric setting with binary data}
As a second example, suppose that $\Xcal = \{0,1\}$, so that estimating $Q_\star(\cdot|z)$ is equivalent to estimating the regression function $p_\star(z) \coloneqq \PPst{X=1}{Z=z}$.
Assuming that this probability is bounded away from 0 and 1, and again applying Pinsker's inequality, we see that, under mild conditions,
\[\tv^2\big(Q_\star^n(\cdot|\Z),Q^n(\cdot|\Z)\big)  \leq \frac{1}{2}\sum_{i=1}^n \kl\Big(Q_\star(\cdot|Z_i),Q(\cdot|Z_i)\Big) \asymp \sum_{i=1}^n \Big(\widehat{p}(Z_i) - p_\star(Z_i)\Big)^2,\]
where $\widehat{p}(z)$ is our estimate of $p_\star(z) = \PPst{X=1}{Z=z}$ based on the unlabeled sample.

Since we are working in a nonparametric setting, 
suppose that we estimate $p_\star(z) = \PPst{X=1}{Z=z}$ via a kernel method, working in a low-dimensional space $\Zcal = \R^k$.
Then standard nonparametric theory ensures that, at ``most'' values $z$, we can achieve error
\[\big(\widehat{p}(z) - p_\star(z)\big)^2 \sim N^{-a_k},\]
where the exponent $a_k$ is a small positive value, depending on both the ambient dimension $k$ and the properties of the function $z\mapsto p_\star(z)$ (e.g.,
smoothness or Lipschitz properties). 
Therefore, we can expect to have 
\[\tv\big(Q_\star^n(\cdot|\Z),Q^n(\cdot|\Z)\big)  \lesssim \sqrt{n \cdot N^{-a_k}},\]
which is vanishing whenever the unlabeled sample size $N$ is sufficiently large relative to the labeled sample size $n$.

\section{Empirical results}\label{sec:empirical}

We next examine the empirical performance of the CPT and CRT on simulated data, and on real data from the Capital Bikeshare system. Code for reproducing all experiments is available on the authors' websites.\footnote{Available at \url{http://www.stat.uchicago.edu/~rina/cpt.html}.}

\subsection{Simulated data: power and error control}
The results of Section~\ref{sec:robust} show that the CPT is more robust than the CRT to errors in the estimated conditional distribution $Q(\cdot|Z)$, when the worst case test statistics $T(\X,\Y,\Z)$ are used. Our first aim here is to provide evidence to validate this result, and to show that this extra robustness is not only exhibited by the worst case test statistic but also for practical and simple choices of $T$. Our second aim is to examine the power of the CPT and CRT to detect deviations from the null hypothesis.

In all of our simulations we set $\alpha = 0.05$ as the desired Type I error rate, and use marginal absolute correlation $T(\X,\Y,\Z) = |\textnormal{Corr}(\X,\Y)|$ as our test statistic. We generate $M=500$ copies of $\X$ under either CPT or CRT. To run the CPT, we use Algorithm~\ref{alg:MC_exchangeable} with $S=50$ steps. All results are shown averaged over 1000 trials.

\subsubsection{Simulations under the null}
First we test whether the CPT and CRT show large increases in Type I error when the conditional distribution estimate $Q(\cdot|Z)$ is incorrect, in a setting where the null hypothesis $H_0:X\independent Y|Z$ holds. 

We will have $X,Y\in\R$ and $Z\in\R^p$ for $p=20$. We first draw independent parameter vectors
\[a,b \sim\mathcal{N}_p(0,\mathbf{I}_p).\]
The variables $(X,Y,Z)$ are then generated as
\[Z\sim \mathcal{N}_p(0,\mathbf{I}_p), \ X|Z\sim Q_\star(\cdot|Z), \ Y|X,Z \sim \mathcal{N}(p^{-1}a^\top Z, 1),\]
where $Q_\star(\cdot|Z)$ will be specified below. (Note that $Y|X,Z$ depends on $Z$ only, since we are working under the null hypothesis that $X\independent Y|Z$.)

Throughout, the estimated conditional distribution of $X|Z$ will be given by $Q(\cdot|Z) = \mathcal{N}(b^\top Z,1)$, but this estimate might not be exactly correct. We will consider several different sources of error in this model:
\begin{enumerate}
\item Nonlinear mean. One source of error comes from assuming a linear relationship between variables where this is in fact not the case.
We choose sample size $n=50$, and  try three different simple examples, taking $Q_\star(\cdot | z) = \mathcal{N}(\mu(z),1)$, where $\mu(\cdot)$ is given by:
\begin{enumerate}
\item Quadratic: $\mu(z)=b^\top z + \theta (b^\top z)^2$,
\item Cubic: $\mu(z)=b^\top z - \theta (b^\top z)^3$,
\item Tanh: $\mu(z)=\tanh(\theta\cdot b^\top z)/\theta$.
\end{enumerate}
In each case, $\theta\geq 0$ is the model misspecification parameter. Note that $\theta=0$ corresponds to the case that $Q(\cdot|Z)=Q_\star(\cdot|Z)$, i.e., the estimate is indeed correct, while larger values of $\theta$ correspond to increasing errors.

\item Coefficients estimated on unlabeled data. Even if the form of the model for $X|Z$ is correct, 
the coefficients $b$ may not be known perfectly. As described earlier, in many practical settings we may have access to ample unlabeled data $(X,Z)$, separate
from our labeled data set of points $(X,Y,Z)$ used to test the hypothesis of conditional independence.
For this setting, we estimate the unknown coefficient vector $b$ with $\widehat{b}$, 
defined as the least-squares estimate using an unlabeled sample $(X^{\textnormal{unlab}}_i,Z^{\textnormal{unlab}}_i),i=1,\dots,N$,
generated independently of the data points $(X_i,Y_i,Z_i)$.
This experiment is repeated for unlabeled sample sizes $N=50,100,\dots,500$. The labeled sample size is given by $n=50$ in each case.
\item Coefficients estimated by reusing the data. Finally, in settings where unlabeled data may not be available,
we may be tempted to estimate the model of $X|Z$ simply using our data points $(X_i,Y_i,Z_i),i=1,\dots,n$. This approach 
is not covered by our theory (since the conditional distribution $Q(X|Z)$ is data-dependent in this case), but it is certainly of practical interest
to see how the method performs in this setting. We test sample sizes $n=50,100,\dots,500$, 
in each case estimating the unknown true coefficient vector $b$ with $\widehat{b}$, which in this case is now
given by the least-squares regression of $X$ on $Z$ trained on the {\em same} data set, $(X_1,Z_1),\dots,(X_n,Z_n)$.
\end{enumerate}

\paragraph{Results} The plots in Figures~\ref{fig:nonlinear_mean} and~\ref{fig:estimate_Q} show the results of these experiments when we have a nonlinear mean, and when we estimate the coefficients using unlabeled data or reusing data, respectively. As the null hypothesis, $H_0:X\independent Y|Z$, is true in all of these experiments, we would hope for the probability of rejection to be close to the nominal level of $\alpha=0.05$, at least when the model misspecification parameter $\theta$ is not too large (for the nonlinear mean setting) or when the unlabeled sample size $N$ or labeled sample size $n$ is not too small (when the model coefficients are trained on unlabeled data or reused data). 

For the nonlinear mean experiments, in Figure~\ref{fig:nonlinear_mean} we see that in many cases the CPT is significantly more robust than the CRT. The $\theta=0$ cases confirm that both tests achieve the nominal Type I error level $\alpha=0.05$ when the assumed distribution $Q$ is correct. As the misspecification parameter
$\theta$ increases (so that the model $Q(\cdot|z)$ that we use for running CPT or CRT, grows farther from the true model $Q_\star(\cdot|z)$),
we see that both methods suffer an inflation of the Type I error level, but for the CPT the excess Type I error is substantially lower than that of the CRT. 

Next, we turn to the setting where the estimated model $Q(\cdot|z)$ is obtained by regressing $X$ on $Z$ using either a separate unlabeled data set,
shown in Figure~\ref{fig:estimate_Q_unlabeled}, or by reusing the same data set, shown in Figure~\ref{fig:estimate_Q_reuse}.
The results are encouraging, showing that, when using unlabeled data, the Type I error is already very close to the nominal level as soon
as the unlabeled sample size $N$ is larger than $n$. When reusing the data, the method in fact appears to be somewhat conservative at smaller
sample sizes $n$---the cause of this phenomenon is an interesting question we hope to study in future work.

\begin{figure}
\centering
\begin{subfigure}{.36\textwidth}
  \centering
  \includegraphics[width=\linewidth]{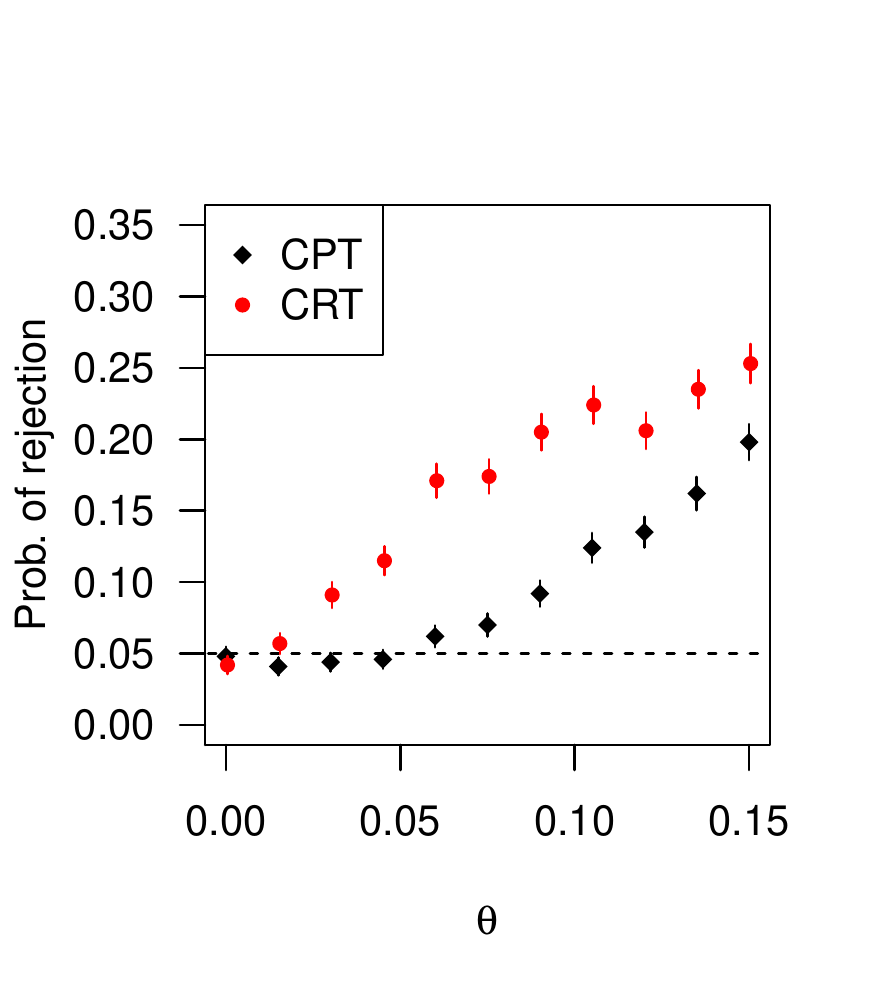}
  \caption{Quadratic}
\end{subfigure}\hspace{-.3in}
\begin{subfigure}{.36\textwidth}
  \centering
  \includegraphics[width=\linewidth]{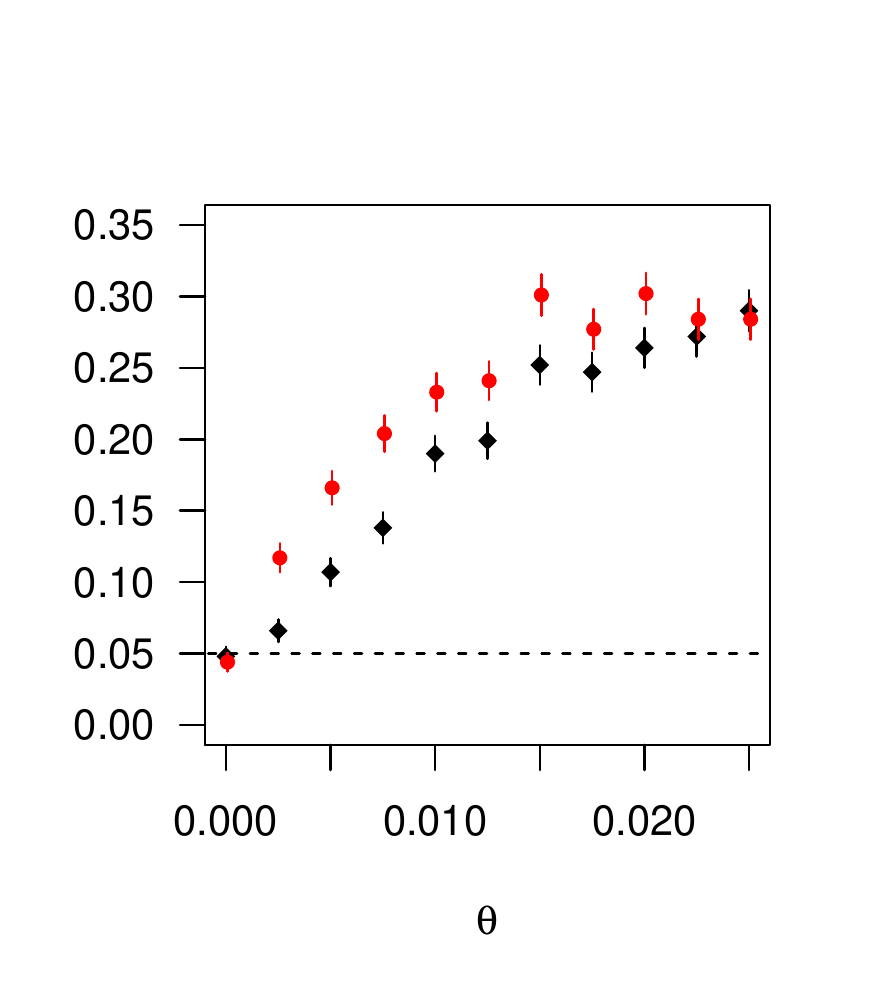}
  \caption{Cubic}
\end{subfigure}\hspace{-.3in}
\begin{subfigure}{.36\textwidth}
  \centering
  \includegraphics[width=\linewidth]{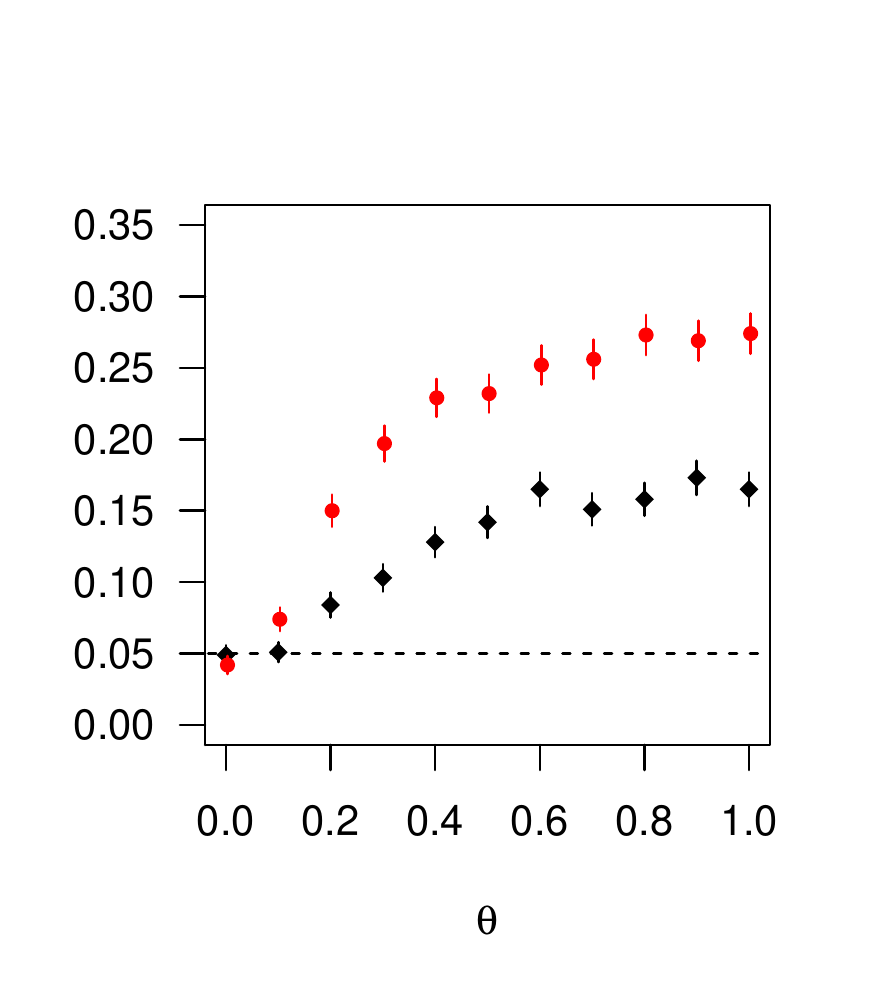}
  \caption{Tanh}
\end{subfigure}
\caption{Simulation results for robustness to misspecification of the mean function. The figures show the probability of rejection (i.e., the Type I error rate), plotted against the model misspecification parameter $\theta$. The plots show the average rejection probability with standard error bars computed over 1000 trials for the CPT and CRT. The dashed line indicates the nominal level $\alpha = 0.05$.}
\label{fig:nonlinear_mean}
\end{figure}

\begin{figure}
\centering
\begin{subfigure}{.4\textwidth}
  \centering
  \includegraphics[width=\linewidth]{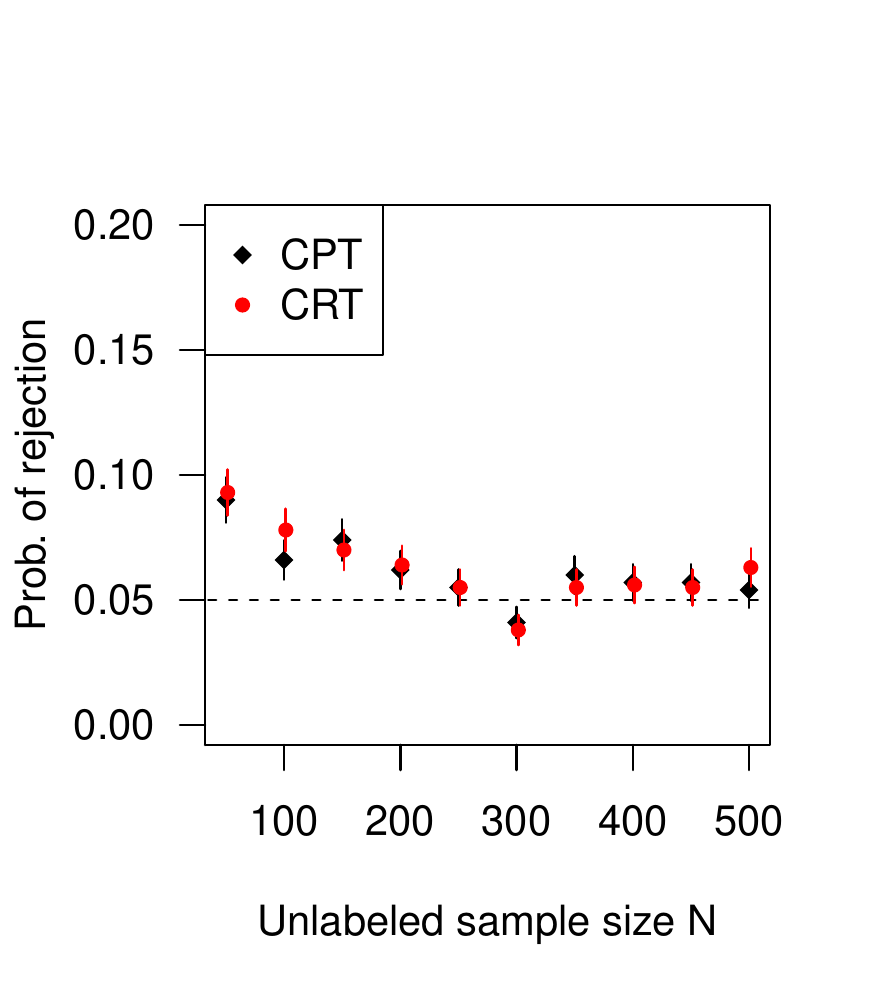}
  \caption{Model trained on unlabeled data}\label{fig:estimate_Q_unlabeled}
\end{subfigure}\hspace{.5in}
\begin{subfigure}{.4\textwidth}
  \centering
  \includegraphics[width=\linewidth]{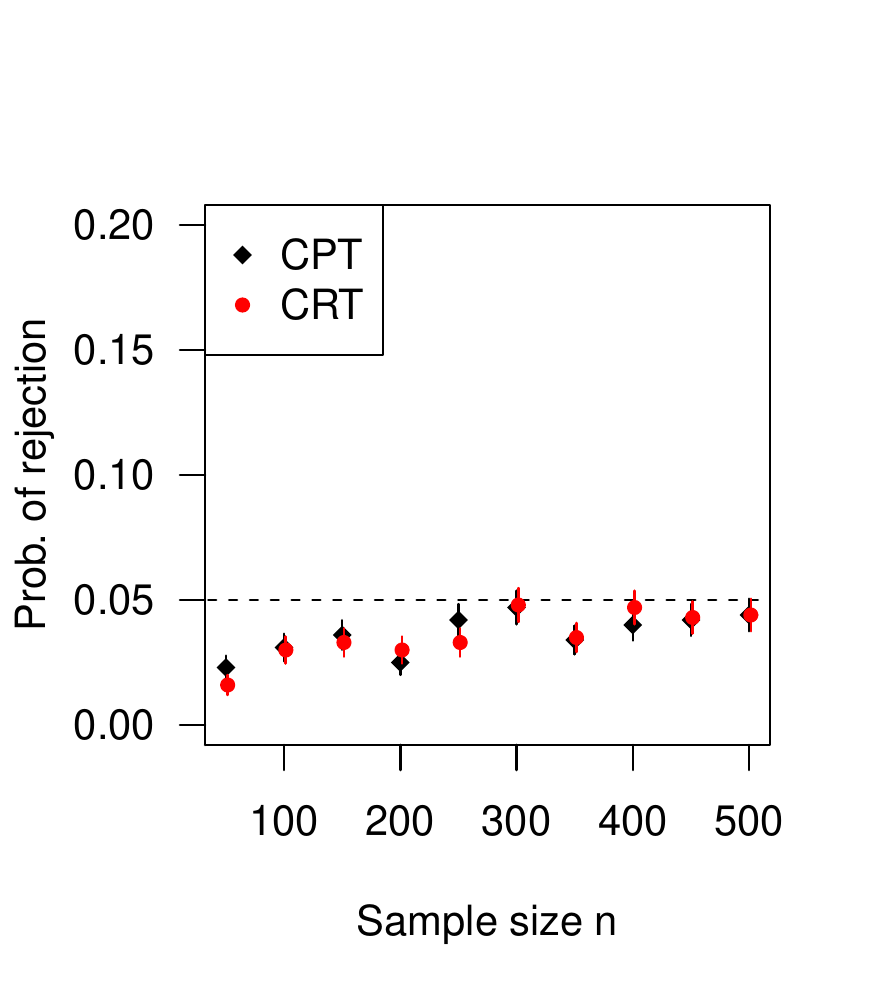}
  \caption{Model trained by reusing data}\label{fig:estimate_Q_reuse}
\end{subfigure}
\caption{Simulation results for robustness to models trained on unlabeled data or by reusing the data. Details as for Figure~\ref{fig:nonlinear_mean}.}
\label{fig:estimate_Q}
\end{figure}

\subsubsection{Simulations under the alternative}\label{sec:sim_alt}
Our final simulation concerns the power of the tests. Here we generate $Z$ as before, and generate $X|Z \sim \mathcal{N}(b^\top Z,1)$, exactly according the assumed distribution $Q(\cdot|Z)$, so that both tests have the nominal Type I error level $\alpha = 0.05$. Unlike the null setting, we now generate $Y|X,Z\sim \mathcal{N}(a^\top Z + c X,1)$. The strength of the signal is controlled by the parameter $c\geq 0$, where $c=0$ corresponds to the null hypothesis being true while larger values of $c$ move farther away from the null. The results, shown in Figure~\ref{fig:power}, reveal that the CPT is slightly less powerful than the CRT across a range of values of $c$, but overall shows fairly similar performance.  Thus there is only a small price to pay for the additional robustness of the CPT.

\begin{figure}
\centering
\includegraphics[width=.4\textwidth]{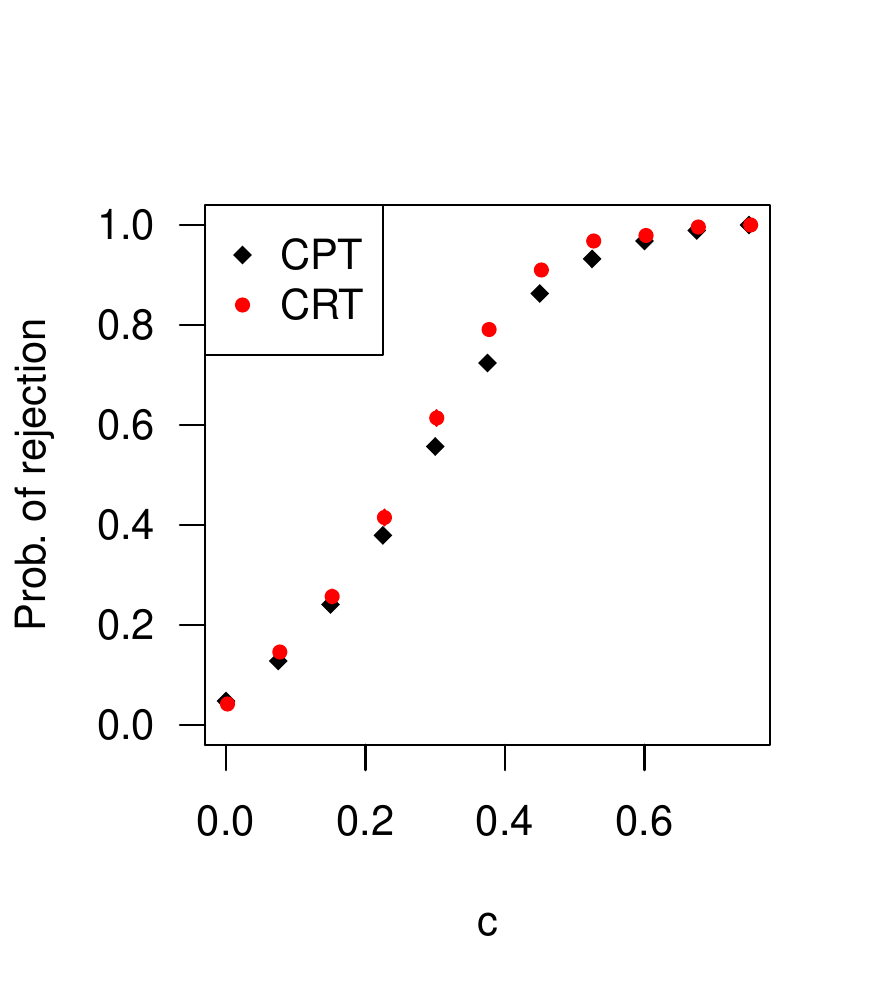}
\caption{Simulation results testing power under the alternative. The figures show the probability of rejection (i.e., the power), plotted against the signal strength parameter $c$. The plots show the average rejection probability with standard error bars computed over 1000 trials for the CPT and CRT. The tests are run at level $\alpha = 0.05$.}
\label{fig:power} 
\end{figure}

\subsection{Simulated data: mixing of the CPT sampler}

\begin{figure}
\centering
\begin{subfigure}{.4\textwidth}
  \centering
  \includegraphics[width=\linewidth]{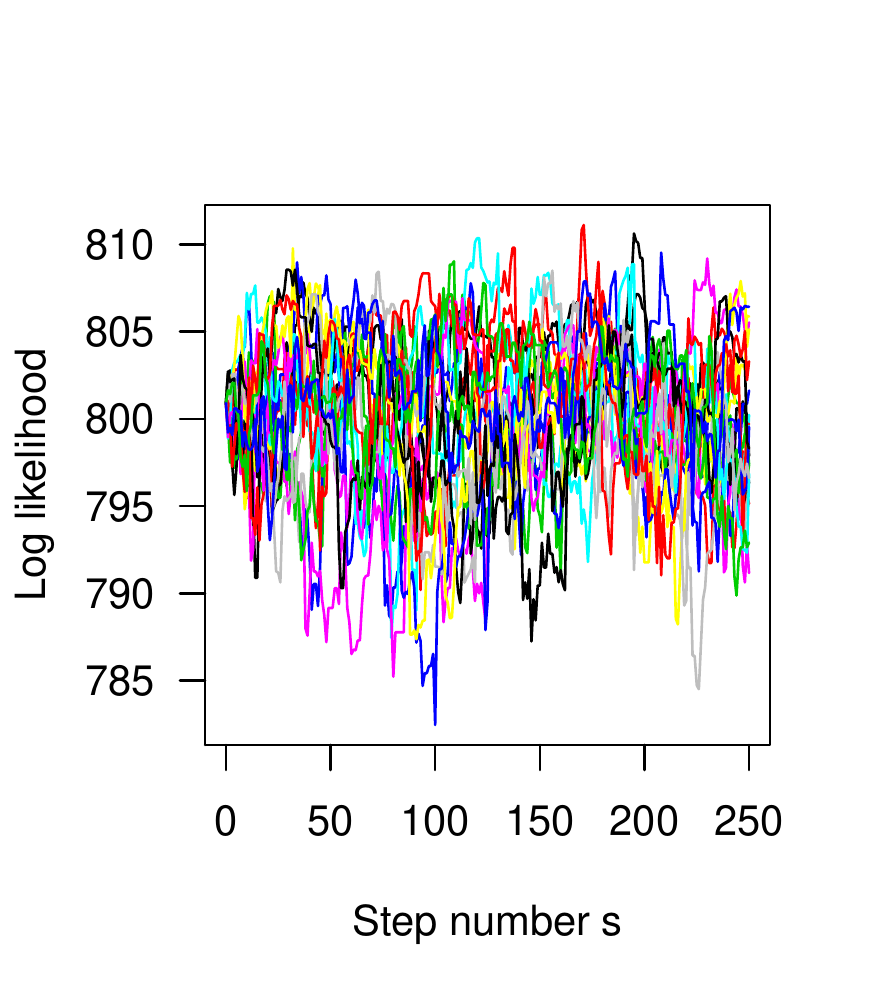}
  \caption{Log-likelihood of $\X^{[s]}$}
\end{subfigure}\hspace{.5in}
\begin{subfigure}{.4\textwidth}
  \centering
  \includegraphics[width=\linewidth]{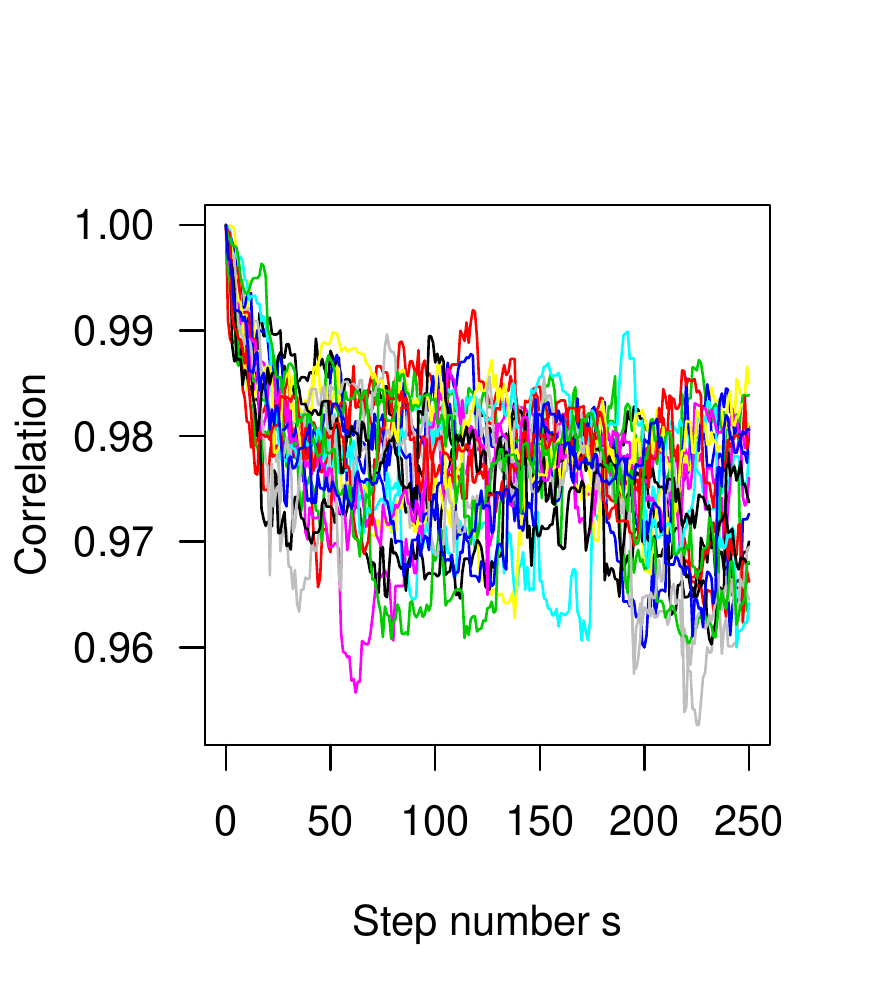}
  \caption{Corr$(\X,\X^{[s]})$}
\end{subfigure}
\caption{Simulation results showing trace plots for the CPT sampler, examining the CPT copy $\X^{[s]}$ at step $s$ of Algorithm~\ref{alg:MC_parallel}.}
\label{fig:traceplots}
\end{figure}

In practice, we cannot implement  the CPT method as defined in~\eqref{eqn:CPT_sample} (unless, of course, the sample size $n$ is so small
that we can simply enumerate all $n!$ possible permutations). Instead, in our experiments,
we use the exchangeable MCMC sampler, defined in Algorithm~\ref{alg:MC_exchangeable}. All of our simulations and real data experiments
implement this sampler with $S=50$, meaning that the Markov chain is run for 50 steps for each new permuted copy $\X^{(m)}$ of the data.
Is this moderate number of steps sufficient to ensure that the chain has mixed well, or are we producing highly correlated data
that will lead to reduced power? To examine this question, we generate one data set, consisting of confounders $Z$ and feature $X$ generated
exactly as in Section~\ref{sec:sim_alt}, and then run the parallel pairwise sampler (Algorithm~\ref{alg:MC_parallel})
independently for 20 trials (i.e., each time initializing at the same original data). At each iteration, setting $\X^{[s]} = \X_{(\Pi^{[s]})}$ to be our current
CPT copy of the original data vector $\X$,  we track the log-likelihood,
 $\sum_{i=1}^n q(X^{[s]}_i|Z_i)$, and the correlation with the original data vector, Corr$(\X,\X^{[s]})$. (Note that, since $X$ is strongly dependent with $Z$,
 it is to be expected that two draws of the data, i.e., $\X$ and $\X^{[s]}$, will necessarily have a high correlation.)
 The trace plots of these two quantities, plotted over $s=0,1,2,\dots,250$ in Figure~\ref{fig:traceplots}, demonstrate that, in this simulation, the Markov chain appears to mix
quickly, within about 50 or 100 iterations. Of course, this will be affected by factors such as the strength of the dependence between $X$ and $Z$, and the sample size $n$.

\subsection{Capital bikeshare data set}
We next implement the CPT and CRT on the Capital Bikeshare data set.\footnote{Data obtained from \url{https://www.capitalbikeshare.com/system-data}.} Capital Bikeshare is a bike sharing program in Washington, D.C., where users may check out a bike from one of their locations and return at any other location. The data set contains each ride ever taken, recording the start time and location, end time and location, bike ID number, and a user type which can be ``Member'' (i.e., purchasing a long-term membership in the system) or ``Casual'' (i.e., paying for one-time rental or a short-term pass). We use the following data:
\begin{itemize}
\item Test data set: all rides taken on weekdays (Monday through Friday) in October 2011. Sample size $n$ = 7,346 rides, after an initial screening step (details below).
\item Training data set (for fitting the conditional distribution $Q(\cdot|Z)$): all rides taken on weekdays in September 2011 and November 2011. Sample size $n_{\textnormal{train}}$ = 149,912 rides.
\end{itemize}
In our experiments, we are interested in determining whether the duration, $X$, of the ride is dependent on various factors $Y$, such as user type (``Member'' or ``Casual''). Of course, the duration of the ride will be heavily dependent on the length of the route, in addition to other factors, and so to control for this we let $Z$ encode both the route, i.e., the start and end locations, as well as the time of day at the start of the ride, since varying traffic might also affect the speed of the ride.  

In order to implement the CPT and CRT, we will use a conditional normal distribution, i.e., $(X|Z=z)\sim \mathcal{N}(\mu(z),\sigma^2(z))$ as an estimate $Q(\cdot|z)$ of $Q_\star(\cdot|z)$.  Before running the CPT or CRT, as an initial screening step we discard any test points for which we do not have a good estimate of the conditional distribution of $X$, keeping only  those test data points where we have ample training data for rides taken along the same route and at similar times of day. The details for fitting $Q(\cdot|Z)$, and for this initial screening step, are given in Appendix~\ref{app:bikeshare_details}. For both the CPT and CRT, we sample $M=1000$ copies of $\X$ to produce the p-value. For the CPT, the Monte Carlo sampler given in Algorithm~\ref{alg:MC_exchangeable} is run with $S=50$ as the number of steps for producing each copy.

\paragraph{Results} We test the null hypothesis $H_0:X\independent Y|Z$ for several different choices of the response $Y$:
\begin{itemize}
\item User type (``Member'' or ``Casual''). We might expect that ``Casual'' users, who are likely to be tourists or infrequent bike riders, may ride at a slower speed.
\item Date, treated as continuous. Since the test data set is taken from the single month October 2011, the date of this month is a continuous variable that acts as a proxy for factors such as weather and the time of sunrise and sunset.
\item Day of the week (Monday through Friday), treated as categorical. Bike riders' behavior may differ on different days of the week, for instance, if rides on Friday are more likely to be leisure rides than the other days of the week.
\end{itemize}
For user type and date, the statistic $T(\X,\Y,\Z)$ that we use is the correlation between the vector $\Y$, and the vector of ride duration residuals after controlling for the effects of $Z$---in other words, the vector with entries $R_i = X_i - \Ep{X\sim Q(\cdot|Z_i)}{X}$. For day of the week, our statistic $T(\X,\Y,\Z)$ is given by 
\[\max_{y\in\{\textnormal{Mon},\dots,\textnormal{Fri}\}} \big|\text{Correlation between $(R_1,\dots,R_n)$ and $(\One{Y_1=y},\dots,\One{Y_n=y})$}\big|.\]

\begin{table}
\begin{center}
\begin{tabular}{|c|c|c|}
\hline
Variable $Y$ &CPT p-value (std.~err.)& CRT p-value (std.~err.)\\\hline
User type&   0.0010 (0.0000) & 0.0010 (0.0000)\\\hline
Date& 0.1146 (0.0032)& 0.1293 (0.0032) \\\hline
Day of week& 0.1980 (0.0037)& 0.2063 (0.0032) \\\hline
\end{tabular}
\caption{p-values obtained from the CPT and CRT for the Capital Bikeshare data. The mean p-value and standard error are calculated from 10 trials of each experiment (the randomness comes from the construction of the copies $\X^{(m)}$ for each test).}
\label{tab:bikeshare}
\end{center}
\end{table}
Table~\ref{tab:bikeshare} shows the resulting p-values for each choice of the variable $Y$. We can see that the CPT and CRT produce nearly identical p-values in all three cases.  We conclude that the user type and duration of ride are dependent, even after controlling for our various confounding variables; on the other hand there is insufficient evidence to reach the same conclusion for the corresponding tests for the date and the day of the week.

\section{Discussion}\label{sec:discussion}
In this work, we have developed a conditional permutation test that modifies the standard permutation test of independence between $X$ and $Y$ in order to account for a known dependence of $X$ on potentially relevant confounding variables $Z$. Our theoretical results prove finite-sample Type~I error control, even when the distribution of $X|Z$ is not known exactly.

We have shown that, empirically, resampling from the set of observed $X$ values preserves better Type~I error control under mild errors in our model, and does not lose much power, in settings where we use intuitive statistics such as correlation between $Y$ and $X$ after regressing out the effects of $Z$. In contrast, our theoretical understanding of Type~I error control covers the worst-case scenario over all possible statistics, and it may be the case that the simple statistics used in practical analyses suffer much less inflation of the Type~I error. We hope to bridge this gap in future work, and also to provide some theoretical insight into the power of the CPT method, as well as to study the efficiency of the Monte Carlo sampler for the CPT and examine whether proposing swaps non-uniformly may improve the speed at which we can obtain copies $\X^{(m)}$ that are not too correlated with each other.

Furthermore, in many applications it might not be possible to estimate the conditional distribution of $X|Z$ independently of the data---if only a small labeled data set $(X,Y,Z)$ is available, with no additional unlabeled data $(X,Z)$ with which to estimate this distribution, we would of course have the option of splitting the data set to use one half for fitting $Q(X| Z)$ and the remaining half to run the CPT, but this would incur substantial loss of both Type I error control and power when the sample size is limited. It is therefore important to consider how the CPT (and the CRT) can retain their validity when the data is used for estimating $Q(X| Z)$ and then reused for testing $H_0:X\independent Y| Z$. It is possible that tools from the selective inference literature may allow us to develop theory towards addressing this question.

Finally, both the CPT and the CRT are based in a setting where it is assumed that modeling $X|Z$ is easy while modeling $Y| X,Z$ is hard---that is, our estimate $Q(\cdot | Z)$ of the conditional distribution $X|Z$ is assumed to be highly accurate, but testing $H_0:X\independent Y| Z$ is a substantial challenge. In contrast, many of the asymptotic tests described in Section~\ref{sec:lit_review} treat the $X$ and $Y$ variables symmetrically when testing $X\independent Y | Z$. Are there settings in which we can construct methods offering finite-sample guarantees in the style of the CPT and CRT while taking a more symmetric approach to this testing problem?

\appendix

\section{Proofs}\label{app:proofs}

\subsection{Proving validity of the sampling mechanisms}

\begin{proof}[Proof of Theorem~\ref{thm:MC}]
This proof consists of simply checking the detailed balance equations for the Markov chain defined by the algorithm.

Let $\mathcal{P}$ be the set of all partitions of $\{1,\ldots,n\}$ into $\lfloor n/2\rfloor$ disjoint pairs. For any $p\in\mathcal{P}$ and any permutations $\pi,\pi'$, we write $\pi\sim_p \pi'$ if $\pi$ can be transformed to $\pi'$ by swapping any subset of the pairs in the partition $p$. For example, if $(i,j), (k,\ell)$ are two of the disjoint pairs in the partition $p$, and $\pi$ and $\pi'$ are related via $\pi' = \pi\circ \sigma_{ij}\circ\sigma_{k\ell}$, then $\pi\sim_p \pi'$ (recall that $\sigma_{ij}$ is the permutation that swaps $i$ and $j$). We note that $\sim_p$ defines an equivalence relation on the set of permutations.

We now compute the transition probability matrix of the Markov chain defined by Algorithm~\ref{alg:MC_parallel}. For ease of notation, for the remainder of this proof, we will condition on $\X_{()},\Y,\Z$ implicitly. In particular, all probabilities $\PP{\cdot}$ or $\PP{\cdot | \cdot}$ should be interpreted as $\PPst{\cdot}{\X_{()},\Y,\Z}$ or $\PP{\cdot | \cdot\, , \X_{()},\Y,\Z}$.

For any permutations $\pi,\pi'$, we have
\[\PPst{\Pi^{[t]}=\pi'}{\Pi^{[t-1]}=\pi} \\= \frac{1}{|\mathcal{P}|}\sum_{p\in\mathcal{P}} \PPst{\Pi^{[t]}=\pi'}{\Pi^{[t-1]}=\pi, \text{ $t$th partition = $p$}},\]
since at each time $t$, Algorithm~\ref{alg:MC_parallel} begins by drawing a partition $p\in\mathcal{P}$ uniformly at random.
Next, given $p$ and $\Pi^{[t-1]}=\pi$, $\Pi^{[t]}$ must satisfy $\Pi^{[t]}\sim_p \pi$ by definition of the next step of the algorithm which can only swap pairs of indices in the partition $p$. By examining the odds ratio defined for each $B_{t,k}$ in~\eqref{eqn:MC_oddsratio}, we see that for any $\pi',\pi''\sim_p \pi$,
\[
\frac{\PPst{\Pi^{[t]}=\pi'}{\Pi^{[t-1]}=\pi, \text{ $t$th partition = $p$}}}{\PPst{\Pi^{[t]}=\pi''}{\Pi^{[t-1]}=\pi, \text{ $t$th partition = $p$}}}\\
=\prod_i \frac{q(X_{(\pi'(i))}|Z_i)}{q(X_{(\pi''(i))}|Z_i)}
=\frac{\PP{\Pi = \pi'}}{\PP{\Pi=\pi''}},
\]
where in the last step we refer to the distribution~\eqref{eqn:distrib_Pi} of the permutation $\Pi$ conditional on $\X_{()},\Y,\Z$. Therefore,
\[\PPst{\Pi^{[t]}=\pi'}{\Pi^{[t-1]}=\pi} = \frac{1}{|\mathcal{P}|}\sum_{p\in\mathcal{P}}\frac{\One{\pi'\sim_p \pi}\cdot \PP{\Pi = \pi'}}{\sum_{\pi''}\One{\pi'' \sim_p \pi}\cdot\PP{\Pi=\pi''}}.\]
Thus, for any $\pi,\pi'$, since $\sim_p$ forms an equivalence relation over permutations, we have
\begin{align*}
&\PP{\Pi = \pi}\cdot \PPst{\Pi^{[t]}=\pi'}{\Pi^{[t-1]}=\pi}\\
&=\frac{1}{|\mathcal{P}|}\sum_{p\in\mathcal{P}}\PP{\Pi = \pi}\cdot \frac{\One{\pi'\sim_p \pi}\cdot \PP{\Pi = \pi'}}{\sum_{\pi''}\One{\pi'' \sim_p \pi}\cdot\PP{\Pi=\pi''}}\\
&=\frac{1}{|\mathcal{P}|}\sum_{p\in\mathcal{P}}\PP{\Pi = \pi'}\cdot \frac{\One{\pi\sim_p \pi'}\cdot \PP{\Pi = \pi}}{\sum_{\pi''}\One{\pi'' \sim_p \pi')}\cdot\PP{\Pi=\pi''}}\\
&=\PP{\Pi = \pi'}\cdot \PPst{\Pi^{[t]}=\pi}{\Pi^{[t-1]}=\pi'}.
\end{align*}
This verifies the detailed balance equations, and so the Markov chain is reversible and has stationary distribution given by~\eqref{eqn:distrib_Pi}. Finally, it is trivial to see that this Markov chain is aperiodic and irreducible when $q(x|z)$ is positive for all $x\in\Xcal$ and $z\in\Zcal$, and so in this case, the stationary distribution is unique.
\end{proof}

\begin{proof}[Proof of Theorem~\ref{thm:MC_exchangeable}]
This result follows directly from the fact that the Markov chain defined in Algorithm~\ref{alg:MC_parallel} is reversible,
as shown in the proof of Theorem~\ref{thm:MC}. This means that, under $H_0$, the permutations $\Pi,\Pi_{\sharp},\Pi^{(1)},\dots,\Pi^{(M)}$ can equivalently be drawn as follows: first draw $\Pi_{\sharp}$ from the distribution~\eqref{eqn:distrib_Pi} conditional on $\X_{()},\Y,\Z$, then draw $\Pi,\Pi^{(1)},\dots,\Pi^{(M)}$ via $M+1$ independent runs of Algorithm~\ref{alg:MC_parallel} for $S$ steps initialized at $\Pi^{[0]}=\Pi_{\sharp}$. Thus $\Pi,\Pi^{(1)},\dots,\Pi^{(M)}$ are i.i.d.~conditional on $\Pi_{\sharp},\X_{()},\Y,\Z$, and are therefore exchangeable.
\end{proof}

\subsection{Proving robust Type~I error control}
\begin{proof}[Proof of Theorem~\ref{thm:robust}]
First we prove the result for the CRT. Let $\check{\X}$ be an additional copy drawn also from $Q(\cdot|\Z)$, independently of $\Y$ and of $\X,\X^{(1)},\dots,\X^{(M)}$. Then, since conditional on $\Y,\Z$ the copies $\X,\check{\X},\X^{(1)},\dots,\X^{(M)}$ are independent, we have
\begin{multline*}\tv\bigg(\Big((\X,\X^{(1)},\dots,\X^{(M)}) | \Y,\Z\Big),\Big( (\check{\X},\X^{(1)},\dots,\X^{(M)}) | \Y,\Z\Big)\bigg) \\
= \tv\Big((\X | \Y,\Z) ,(\check{\X}| \Y,\Z) \Big)= \tv\big(Q_\star^n(\cdot|\Z),Q^n(\cdot|\Z)\big).\end{multline*}
Now let $A_\alpha\subseteq (\Xcal^n)^{M+1}$ be defined as
\[A_\alpha := \biggl\{(\x,\x^{(1)},\dots,\x^{(M)}) : \frac{1+\sum_{m=1}^M\One{T(\x^{(m)},\Y,\Z)\geq T(\x,\Y,\Z)}}{1+M}\leq \alpha\biggr\},\]
i.e., the set where we would obtain a p-value $p\leq \alpha$. Then
\begin{align*}&\PPst{p\leq \alpha}{\Y,\Z} = \PPst{(\X,\X^{(1)},\dots,\X^{(M)}) \in A_\alpha}{\Y,\Z} \\
&\leq  \PPst{(\check{\X},\X^{(1)},\dots,\X^{(M)}) \in A_\alpha}{\Y,\Z} \\
&\hspace{1in}+ \tv\bigg(\Big((\X,\X^{(1)},\dots,\X^{(M)}) | \Y,\Z\Big),\Big( (\check{\X},\X^{(1)},\dots,\X^{(M)}) | \Y,\Z\Big)\bigg) \\
&=  \PPst{(\check{\X},\X^{(1)},\dots,\X^{(M)}) \in A_\alpha}{\Y,\Z} +  \tv\big(Q_\star^n(\cdot|\Z),Q^n(\cdot|\Z)\big).\end{align*}
Finally, since $\check{\X},\X^{(1)},\dots,\X^{(M)}$ are clearly \iid~after conditioning on $\Y,\Z$, and are therefore exchangeable, by definition of $A_\alpha$ we must have
\[ \PPst{(\check{\X},\X^{(1)},\dots,\X^{(M)}) \in A_\alpha}{\Y,\Z} \leq \alpha,\]
proving the desired bound for the CRT.

Next we turn to the CPT, for which the analysis is more complicated since the $\X^{(m)}$'s depend on the observed values in the vector $\X$. We will use the fact that,
\begin{equation}\label{eqn:tv_fact}\begin{tabular}{c}\text{For any $(U,V)$ and $(U',V')$, if $(V| U=u) \eqd (V'| U' = u)$ for any $u$,}\\\text{ then $\tv\Big((U,V),(U',V')\Big) = \tv(U,U')$.}\end{tabular}\end{equation}
Let $\check{\X}$ be drawn from $Q(\cdot|\Z)$, independently of $\Y$, and let $\check{\X}^{(1)},\dots,\check{\X}^{(M)}$ be draws from the CPT when we sample from the values of $\check{\X}$ instead of $\X$. That is, independently for each $m=1,\dots,M$, we draw
\[\check{\X}^{(m)}  = \check{\X}_{(\check{\Pi}^{(m)})}\text{ \ where \ }\PPst{\check{\Pi}^{(m)} = \pi}{\check{\X}_{()},\Y,\Z} \propto \ q^n(\check{\X}_{(\pi)}| \Z) ,\]
where $\check{\X}_{()}$ and $\check{\X}_{(\pi)}$ are defined analogously to $\X_{()}$ and $\X_{(\pi)}$ from Section~\ref{sec:CPT}. 
Next, by comparing to the CPT sampling mechanism~\eqref{eqn:CPT_sample_alt}, we observe that the $\check{\X}^{(m)}$'s, conditional on $\check{\X}$, are generated with the same mechanism as the $\X^{(m)}$'s conditional on $\X$. In other words, for any $\x\in\Xcal^n$, we have
\[\Big(\big(\check{\X}^{(1)},\dots,\check{\X}^{(M)}\big)| \check{\X}=\x,\Y,\Z\Big) \eqd \Big(\big(\X^{(1)},\dots,\X^{(M)}\big)| \X=\x,\Y,\Z\Big) .\]
We can verify that
the same equality in distribution
holds if we instead use the exchangeable sampler (Algorithm~\ref{alg:MC_exchangeable}) with some choice $S\geq 1$ of the number of steps.

In either case, then, applying~\eqref{eqn:tv_fact} we have
\begin{multline*}\tv\bigg(\Big((\X,\X^{(1)},\dots,\X^{(M)}) | \Y,\Z\Big),\Big( (\check{\X},\check{\X}^{(1)},\dots,\check{\X}^{(M)}) | \Y,\Z\Big)\bigg) \\
= \tv\Big((\X | \Y,\Z) ,(\check{\X}| \Y,\Z) \Big)= \tv\big(Q_\star^n(\cdot|\Z),Q^n(\cdot|\Z)\big).\end{multline*}
From this point on, we proceed as for the CRT---we have
\[\PPst{p\leq \alpha}{\Y,\Z}
\leq  \PPst{(\check{\X},\check{\X}^{(1)},\dots,\check{\X}^{(M)}) \in A_\alpha}{\Y,\Z} +  \tv\big(Q_\star^n(\cdot|\Z),Q^n(\cdot|\Z)\big),\]
and since $\check{\X},\check{\X}^{(1)},\dots,\check{\X}^{(M)}$ are exchangeable after conditioning on $\Y,\Z$, we see that $\PPst{(\check{\X},\check{\X}^{(1)},\dots,\check{\X}^{(M)}) \in A_\alpha}{\Y,\Z} \leq \alpha$, proving the desired bound for the CPT (with permutations drawn either i.i.d.~as in~\eqref{eqn:CPT_sample_alt}, or from the exchangeable
sampler given in Algorithm~\ref{alg:MC_exchangeable}).
\end{proof}

\begin{proof}[Proof of Theorem~\ref{thm:robust_lowerbd}]
For convenience we will write
\[\tv = \tv\big(Q_\star^n(\cdot|\Z),Q^n(\cdot|\Z)\big)\]
throughout this proof. First, by a standard property of the total variation distance, there exists a subset $A(\Z)\subseteq\Xcal^n$ such that
\[\Pp{Q_\star^n(\cdot|\Z)}{\X\in A(\Z)| \Z} =  \Pp{Q^n(\cdot|\Z)}{\X\in A(\Z)| \Z} + \tv.\]
Fix any $M\geq 2$, and define
\[\alpha_0(\Z) := \Pp{Q^n(\cdot|\Z)}{\X\in A(\Z)| \Z}, \quad \alpha(\Z) := \alpha_0(\Z) + 0.5\sqrt{\frac{\log(M)}{M}}.\]
Now, by definition of the setting and the CRT, we know that conditional on $\Z$, we have $\X\sim Q_\star^n(\cdot|\Z)$ and independently, $\X^{(1)},\dots,\X^{(M)}\sim Q^n(\cdot|\Z)$. Therefore, 
\[\Big(\One{\X\in A(\Z)} | \Y,\Z\Big) \sim\textnormal{Bernoulli}\Big(\alpha_0(\Z) + \tv\Big),\]
and independently,
\[\bigg(\sum_{m=1}^M \One{\X^{(m)}\in A(\Z)} | \Y,\Z\bigg) \sim \textnormal{Binomial}(M,\alpha_0(\Z)).\]

We will work with the statistic $T(\X,\Y,\Z) = \One{\X\in A(\Z)}$.  We have
\begin{align}
\notag&\PPst{p\leq \alpha(\Z)}{\Y,\Z}\\
\notag&= \PPst{\frac{1+\sum_{m=1}^M \One{T(\X^{(m)},\Y,\Z) \geq T(\X,\Y,\Z)}}{1+M}\leq \alpha(\Z)}{\Y,\Z}\\
\notag&\geq \PPst{\X\in A(\Z)\text{ and }\sum_{m=1}^M \One{\X^{(m)}\in A(\Z)} \leq \alpha(\Z)\cdot (M+1)-1}{\Y,\Z}\\
\notag&=\Big(\alpha_0(\Z) + \tv\Big) \cdot \PPst{\textnormal{Binomial}(M,\alpha_0(\Z)) \leq  \alpha(\Z)\cdot (M+1)-1}{\Z}\\
\label{eqn:bennett_step}&\geq  \alpha(\Z) + \tv -0.5\sqrt{\frac{\log(M)}{M}} -   \PPst{\textnormal{Binomial}(M,\alpha_0(\Z)) >  \alpha(\Z)\cdot (M+1)-1}{\Z},
\end{align}
where the last step holds by definition of $\alpha(\Z),\alpha_0(\Z)$, and the fact that $ \alpha_0(\Z) + \tv\leq 1$. Finally, it suffices to bound this binomial probability. By Bennett's inequality, writing $h(u) = (1+u)\log(1+u)-u$, for any $t \in [0,1]$ we have
\begin{align}
\notag&\PP{\textnormal{Binomial}(M,t) >  \left(t + 0.5\sqrt{\frac{\log(M)}{M}}\right)\cdot (M+1)-1}\\
\notag&=\PP{\textnormal{Binomial}(M,t) - Mt >  t + 0.5\sqrt{\frac{\log(M)}{M}}\cdot (M+1)-1}\\
\notag&\leq \exp\left\{ - Mt(1-t)\cdot  h\left(\frac{t + 0.5\sqrt{\frac{\log(M)}{M}}\cdot (M+1)-1}{Mt(1-t)}\right)\right\}\\
\label{eqn:binomial_h}&\leq \exp\left\{ - \frac{M}{4} h\left(\frac{0.5\sqrt{\frac{\log(M)}{M}}\cdot (M+1)-1}{M/4}\right)\right\},
\end{align}
where the last step holds since $h$ is an increasing function, while $c\mapsto c \cdot h(a/c)$ is decreasing in $c>0$, for any $a>0$, and $t(1-t)\leq 1/4$.

Finally, as $\eps\rightarrow 0$, we have $h(\eps) = \eps^2/2 + O(\epsilon^3)$, so as $M \rightarrow \infty$ we have
\begin{multline*}\exp\biggl\{ - \frac{M}{4} h\biggl(\frac{0.5\sqrt{\frac{\log(M)}{M}}\cdot (M+1)-1}{M/4}\biggr)\biggr\} = \exp\left\{-\frac{1}{2}\log(M) + o(1)\right\} \\= \frac{1}{\sqrt{M}} = o(1)\cdot  0.5\sqrt{\frac{\log(M)}{M}}.\end{multline*}
Returning to~\eqref{eqn:bennett_step}, we see that
\[\PPst{p\leq \alpha(\Z)}{\Y,\Z} \geq \alpha(\Z) + \tv -\sqrt{\frac{\log(M)}{M}}\cdot 0.5(1+o(1)).\]
More concretely, for any $M\geq 2$ we can verify numerically that the quantity in~\eqref{eqn:binomial_h} is bounded by $2\sqrt{\frac{\log(M)}{M}}$, which shows that the term $0.5(1+o(1))$ above can be replaced with $2.5$ for any $M\geq 2$.
\end{proof}

\section{Details for bikeshare data experiment}\label{app:bikeshare_details}

We will write $Z=(Z_{\textnormal{route}},Z_{\textnormal{time}})$, where the route encodes both the start and end locations and is treated as categorical. 

To estimate a conditional distribution $Q(\cdot|Z)$, we assume that $X|Z$ is normally distributed, and we fit the conditional mean and variance on the training data by grouping rides according to their route and taking a Gaussian kernel over their start time: for any $z = \big(z_{\textnormal{route}},z_{\textnormal{time}}\big)$,
\[\widehat{\mu}(z) = \sum_i \frac{ w(z,Z^{\textnormal{train}}_i)}{\sum_{i'} w(z,Z^{\textnormal{train}}_{i'})} \cdot X^{\textnormal{train}}_i, \quad \widehat{\sigma}^2(z) = \sum_i  \frac{ w(z,Z^{\textnormal{train}}_i)}{\sum_{i'} w(z,Z^{\textnormal{train}}_{i'})}\cdot (X^{\textnormal{train}}_i)^2 - \big(\widehat{\mu}(z) \big)^2,\]
where the weights are given by grouping observations by route and applying a Gaussian kernel to the time, i.e.
\[w(z,Z^{\textnormal{train}}_i) =  \One{(Z^{\textnormal{train}}_i)_{\textnormal{route}} = z_{\textnormal{route}}} \cdot \exp\big\{ - ((Z^{\textnormal{train}}_i)_{\textnormal{time}} - z_{\textnormal{time}})^2/(2h^2)\big\}\]
for a bandwidth $h$ of 20 minutes. Time of day is on a continuous 24 hour clock, that is, if $z_{\textnormal{time}}$ = 11:00pm and $(Z^{\textnormal{train}}_i)_{\textnormal{time}}$ = 1:00am then the difference between them is two hours, not 22 hours. 

Our conditional distribution estimate $Q(\cdot|Z)$ is then given by
\[(X|Z=z) \sim  \mathcal{N}\Big(\widehat{\mu}(z) , \widehat{\sigma}^2(z)\Big).\]

However, since the popularity of various routes and different times of day varies widely, there are some values $z$ where our estimate of the conditional mean and variance of $X$ is unreliable due to scarce data. To check this, for any $z$ we define
\[N(z) = \sum_i w(z,Z^{\textnormal{train}}_i) ,\]
where a larger $N(z)$ means that there are a larger number of rides in the training data that were taken along the same route $z_{\textnormal{route}}$, and at a time of day similar to $z_{\textnormal{time}}$. For the test data, we then keep only those data points $(X_i,Y_i,Z_i)$ for which $N(Z_i)\geq 20$. Since this screening step uses the value of $Z_i$ but not the value of $X_i$, the $X_i$'s are still unobserved even after screening, and their distribution conditional on $Z_i$ is unchanged; therefore the CPT and CRT tests are valid even on this screened data.

\subsection*{Acknowledgements}
R.F.B.~was partially supported by the NSF via grant DMS-1654076 and by an Alfred~P.~Sloan fellowship. T.B.B.~and R.J.S.~were supported by an EPSRC Programme grant.  R.J.S.~was also supported by an EPSRC Fellowship and a grant from the Leverhulme Trust. The authors would like to thanks the Isaac Newton Institute for Mathematical Sciences for its hospitality during the programme Statistical Scalability which was supported by EPSRC Grant Number: LNAG/036, RG91310. The authors thank Samir Khan for help implementing code for our algorithms.

\bibliographystyle{plainnat}
\bibliography{bib}

\end{document}